\documentclass[10pt]{article}

\usepackage[utf8]{inputenc}

\usepackage[english]{babel}
\usepackage{amsmath}
\usepackage{amssymb}
\usepackage{dsfont} 
\usepackage{amsthm}
\usepackage[T1]{fontenc}
\usepackage[utf8]{inputenc} 
\usepackage{tikz}
\usepackage{verbatim}
\usetikzlibrary{arrows}
\usepackage{float}
\usepackage{subcaption}
\usepackage{caption}
\usepackage{array}
\usepackage{bm}
\usepackage[labelfont=bf]{caption}
\usepackage{enumerate}
\usepackage[final]{pdfpages}
\usepackage{graphicx}
\usepackage{epigraph}
\usepackage{color}
\usepackage{framed}
\usepackage{lipsum} 
\usepackage{geometry} 
\usepackage{marginnote} 

\definecolor{myRed}{rgb}{1,0,0}
\definecolor{myBlue}{rgb}{0,0,1}
\definecolor{myYellow}{rgb}{0,1,0}
\tikzstyle{red}=[thick, myRed, opacity=1]
\tikzstyle{blue}=[thick, myBlue, opacity=1]
\tikzstyle{yellow}=[ thick, myYellow, opacity=1]
\tikzset{edge/.style = {->,> = latex'}}

\newcommand{\path}{\vec{p}}
\newcommand{\supp}{\mathrm{Supp}}

\newcolumntype{L}{>{\arraybackslash}m{3cm}}
\newcommand{\eps}{\varepsilon}

\newcommand{\normal}{\mathcal{N}}

\renewcommand{\P}{\mathbb{P}}
\newcommand{\E}{\mathbb{E}}
\newcommand{\V}{\mathbb{V}}

\newcommand{\N}{\mathbb{N}}
\newcommand{\R}{\mathbb{R}}

\newcommand{\ind}{\mathds{1}}

\renewcommand{\ij}{(i,j)}
\newcommand{\rP}{\vec{\mathcal{P}}}

\def\var#1{\mathbb{V}\left[#1\right]}

\newtheorem{thm}{Theorem}
\newtheorem{defn}{Definition}
\newtheorem{lemma}[thm]{Lemma}

\newtheorem{ex}[thm]{Example}

\newtheorem{corollary}[thm]{Corollary}

\newtheorem{ass}[thm]{Assumption}
\newtheorem{remark}{Remark}

\definecolor{chocolate}{rgb}{0.48, 0.25, 0.0}
\definecolor{berry}{rgb}{0.78, 0.31, 0.39}

\usepackage[normalem]{ulem}
\usepackage{hyperref}

\title{Estimations of means and variances in a Markov  linear model. } 

\author
{Abraham Gutierrez, Sebastian M\"uller
\\
\today
}

\date{}

\begin{document}


\maketitle

%
\begin{abstract}
Multivariate regression models and ANOVA are probably the most frequently applied methods of all statistical analyses. We study the case where the predictors are qualitative variables, and the response variable is quantitative. In this case, we propose an alternative to the classic approaches that does not assume homoscedasticity but assumes that a Markov chain can describe the covariates' correlations. 

This approach transforms the dependent covariates using a change of measure to independent covariates. The transformed estimates allow a pairwise comparison of the mean and variance of the contribution of different values of the covariates. We show that under standard moment conditions, the estimators are asymptotically normally distributed. 

We test our  method with data from simulations and apply it to several classic data sets.



\end{abstract}


\section{Introduction}

We propose a linear model for qualitative predictors, $X^{(1)},\ldots, X^{(m)}$ and a quantitative response variable $Y$. Our approach allows us to model different variances for each category and to consider arbitrary  error terms. It weakens the assumptions of independent predictors of a previous work \cite{GuMu:19}, in the following way: the value of $X^{(k+1)}$ conditioned on $X^{(k)}$ is independent of $X^{(k-1)}$, for all reasonable choices of $k$.  

Our approach is in a ``probabilistic spirit'' since we use Markov chains and a change of measure similar to a  Girsanov transform to model the correlation between the predictors and use classic results on random walks to study the asymptotic behavior of the resulting estimators.

The approach is probably best understood with a comparison to standard linear models.
Let us assume that there are two categorical predictors or groups  $X^{(1)}$ and $X^{(2)}$  and a quantitative response $Y$. We assume that each predictor $X^{(i)}$ can take values in $\{c_{1},\ldots, c_{k_{i}}\}, i\in \{1,2\}$, and encode the values of these variables using the vectors\footnote{In contrast to the standard encoding of categorical variables in linear models with $k_{i}-1$ dummy variables, we choose $k_{i}$ dummy variables.}
\begin{equation}
\widehat X^{(i)}=( \ind\{X^{(i)}=c_{1}\}, \ldots, \ind\{X^{(i)}=c_{k_{i}}\}).
\end{equation}
Let $\alpha^{(i)} \in \R^{k_{i}}$ be real vectors and $\eps^{(i)}$ be random vectors taken values in $\R^{k_{i}}$; we assume that $\E[ \eps^{(i)}]=0$, but allow the variances of each component to be different and denote $\sigma_{j}^{(i)}=\V[ \eps^{(i)}_{j}]$.
We assume the following linear relationship between the two categorical predictors and the response variable:
\begin{equation}\label{eq:linearRelationshipDAG}
y=  (\alpha^{(1)}+ \eps^{(1)}) \widehat x_{1}^{t} +   (\alpha^{(2)}+ \eps^{(2)}) \widehat x_{2}^{t}.
\end{equation}
In a standard linear regression model or ANOVA  the linear relationship between predictors and the response variable looks like
\begin{equation}\label{eq:linearRelationshipRegression}
y=  \alpha^{(1)} \widehat x_{1}^{t} +   \alpha^{(2)} \widehat x_{2}^{t} + \eps.
\end{equation}
Note that here $\eps$ is a random error that it supposed to be centered. In particular, the classic models assume homoscedasticity, i.e., that the contribution to the random error is the same for each category. When this assumption is violated, it is typically referred to as heteroscedasticity. As heteroscedasticity is not uncommon in practice there has been an extensive research in various fields of applications on how to generalize linear models in this respect. We refer to \cite{rosopa} for an overview on impacts of heteroscedasticity on the quality of estimations and on mitigating it effects. 

 We propose a modeling of dependencies between different predictors using Markov chains and a pathwise approach. Since the predictors are categorical variables, every individual of the sample corresponds to a path. For instance, in Figure \ref{fig:path} the bold path corresponds to an individual with $X^{(1)}=X^{(2)}=1$. A Markov chain now gives the dependencies between the different categories. The probability of choosing category $i$ in predictor one is $p_{s, 1i}$; corresponding to the initial measure of a Markov chain. Now, the probability of choosing $j$ in predictor two conditioned on having chosen $i$ in predictor one is given by $p_{1i,2j}$. Again, each path corresponds to the categories of an individual. If the transition probability is those of a  uniform distribution, the categories of the individuals are independent. 

In the case of non-uniform transition, estimations for the contribution of specific factors can be biased. For instance, in the situation of Figure \ref{fig:path} assume that factor $1$ has no contribution, but the category in factor $2$ has a much higher contribution than the other categories of this factor. Now, if the probability that category $1$ of factor $1$ is chosen together with category  $1$ of factor $2$ with much higher probability than the other categories in factor $1$, the contribution of category $1$ of factor $1$ is generally overestimated if one works under the assumption of independence of the factors.  This situation corresponds to multicollinearity  in linear models.  The correlation between predictor variables makes estimation and interpretation of the models more complicated than in the independent case, since  it is difficult to disentangle the impacts of different predictor variables on the response.

We solve this problem in transferring the non-uniform setting to a uniform setting using a discrete Girsanov transform, see Lemma \ref{lemma:MeasureChange}. This transforms quantifies the biases introduced by the variables' dependence. Here, the martingale measure in the standard Girsanov transform corresponds to the setting of uniform transition or independence of the predictors. In other words, the Girsanov transform allows us to quantify the dependence of the predictors and change the measure to a uniform or independent setting.
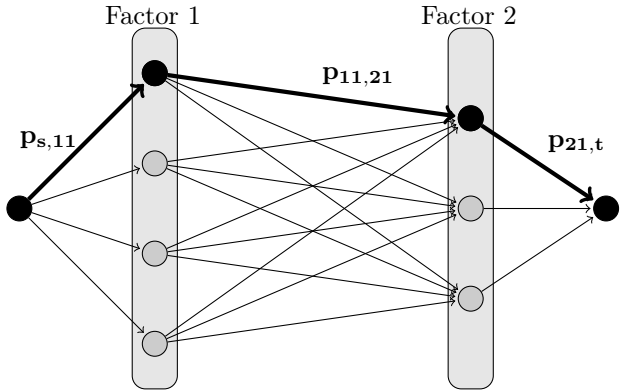
\begin{figure}
\centering
\tikzset{
node_style/.style={circle,draw=black,fill=black!20!, minimum size=0pt},
node_stylethick/.style={circle,draw=black,fill=black, minimum size=0pt},
edge_style/.style={draw=black,  thin, ->},
edge_stylethick/.style={draw=black,  ultra thick, ->},
label/.style={sloped,above}
}
\begin{tikzpicture}[shorten >=1pt,
   auto,
   node distance=2cm and 10cm,
   scale=0.6]

   \draw [rounded corners=5, fill=black!10!] (2.5,-7) rectangle ++(1,8) node [midway] {};
   \node[text width=1.5cm] at (3.15,1.3) {Factor 1};
    \node[text width=1.5cm] at (10.15,1.3) {Factor 2};
        
   \draw [rounded corners=5, fill=black!10!] (9.5,-7) rectangle ++(1,8) node [midway] {};

     \node[node_stylethick] (s) at (0,-3) {};
     
     \foreach \y/\t in {1/1, 2/1, 3/1, 4/1}{
     \node[node_style]  (v\y) at (3, -2*\y+2) {};
     \draw [edge_style]  (s) edge  (v\y);}
       
       \draw [edge_stylethick]  (s) edge  (v1);
       \path (s) to node [yshift=1em] {$\mathbf{p_{s,11}}$} (v2);
	
 \node[node_stylethick] (v1) at (3,0) {};
       
     \foreach \y/\t in {1/2, 2/2, 3/2}
     \node[node_style]  (w\y) at (10, -2*\y+1) {};
     \foreach  \y in {1,2,3,4} 
        \foreach  \w in {1,2,3} 
           \draw [edge_style]  (v\y) edge  (w\w);
           
           \draw [edge_stylethick]  (v1) edge  (w1);
           \path (v1) to node [yshift=0em] {$\mathbf{p_{11,21}}$} (w1);
            \node[node_stylethick] (w1) at (10,-1) {};
           
     \foreach \y/\t in {1/3}
     \node[node_stylethick]  (x\y) at (13, -3) {};
     \foreach  \y in {1,2,3} 
        \foreach  \w in {1} 
           \draw [edge_style]  (w\y) edge  (x\w);
         
          \draw [edge_stylethick]  (w1) edge  (x1);
           \path (w1) to node  {$\mathbf{p_{21,t}}$} (x1);

    \end{tikzpicture}
\caption{An illustration of a path passing through different workstations in a production network.}\label{fig:path}
\end{figure}

\subsection{Motivation}
The proposed model applies in a very general setting. However, it is motivated by real-world applications in the quality control of parallel production networks, see also the previous work \cite{GuMu:19}, and we think it might be useful to have this particular use case in mind.

In parallel production networks, the quality of the end product depends on various intermediate production steps.  Therefore, a recent challenge is to achieve a level of visibility into the production flows that allow us to optimize throughput by guaranteeing at the same time given quality standards. 

We consider a  production network consisting of several workstations; see Figure \ref{fig:path}. Each workstation is a parallel configuration of machines performing the same kind of tasks on a given part. Parts move from one workstation to another, and at each workstation, a part is assigned randomly to a machine. We assume that the production network is \emph{acyclic}, i.e., a part does not return to a workstation where it previously received service. Furthermore, we assume that the quality of the end product is \emph{additive}, that is, the final quality is the sum of the machines' quality contributions along the production path. Separate latent random variables model the contribution of each machine. Note that the product's quality can be replaced by the time it takes to perform specific tasks. In this language, our main result is constructing estimators that allow pairwise and multiple comparisons of the means and variances of different machines in the same workstation. These comparisons then may lead to the identification of unreliable machines.  

In Section \ref{sec:examples}, we treat examples of different kinds to illustrate that the method applies to various kinds of use cases.




\subsection{Related work}
Our situation of the covariates falls into the framework of graph-structured covariates, e.g., see \cite{LiLi:10}. In contrast to \cite{LiLi:10} the ``dependency structure'' in our case is acyclic which allow our more direct solution. 
Multifactor experimental designs,  e.g., see \cite{DrPu:96, Se:18},  are other alternatives in the literature to estimate the mean differences; however mostly rely on homoscedasticity.  They are usually used in the context of statistically planned experiments, which consists of a few experimental runs to obtain data on the product characteristics. If the number of observations for each setting  is sufficiently high and under further conditions described in \cite[Section 4]{DrPu:96} these methods allow a comparison of the variances, too. We also want to mention the connection to critical paths analysis, e.g., see ~\cite{BoGeWeAr:10, Schu:05}.  These methods allow to find critical paths in acyclic networks, however, they seem  suited to compare nor estimate differences in mean and variances.

\subsection{Outline}
In Section \ref{sec:model}, we define the model using Markov chains and directed acyclic graphs. In Section \ref{section:UniformTransition} we summarize the results in the case of independent covariates from \cite{GuMu:19}. Section \ref{sec:Markov} contains the main results. Lemma \ref{lemma:MeasureChange} describes the Girsanov transform. In Theorems \ref{thm:markovlln}, we prove asymptotic consistency of the estimators for mean and variance. In Section \ref{sec:asympt}, their asymptotic distribution is given for the case where the correlations are known in Theorem \ref{thm:Qknown}, and where the correlations are unknown in Theorem  \ref{thm:asymptUnknownQ}.
Finally, we present several examples in Section \ref{sec:examples}.

\section{The model}\label{sec:model}
We use a  directed acyclic graph (DAG) to describe the dependencies of the covariates. 
 A DAG  is a finite directed graph with no directed cycles. It consists of a finite vertex set $V$ and a finite set of directed edges $E=\{ (v,w): v,w \in V, v\neq w\}$. In our setting the DAG contains two special vertices: a source $s$ and a sink $t$.  We are interested in the paths from the source to the sink in this graph.  We denote a path $\vec{p}$ in the DAG as $\vec{p} = (p_{0},p_{1},\ldots, p_{c}, p_{c+1})$ where $p_{0}=s$ and $p_{c+1}=t$ and $(p_{i},p_{i+1})\in E$. We define $\vec{p}[j]:=p_{j}$.  We refer to Figure \ref{fig:DAG1} for an illustration and to \cite{BaGu:09} for more details on directed graphs.

We assume that at each step $1\leq i \leq c$ the path $\vec{p}$ has  $r_{i}\leq r$  different choices and the nodes in each column are always numerated starting with $1$.  The possible choices of a path can therefore be modeled through an $r\times c$ matrix.   More precisely, given a path $\vec{p}$, we associate an $r \times c$ binary matrix $V_{\vec{p}}$ that has $1$'s 
only in the nodes visited by the path:
$$
V_{\vec{p}} := (V_{\vec{p}}(i,j))_{i\in[r],\, j\in [c]},
$$
where we denote $[k]:=\{1,2,\ldots, k\}$ for an integer $k$.
We call $V_{\vec{p}}$ the {indicator matrix} of the path $\vec{p}$.

Each path contains exactly one node of each column. This paper aims to study differences among nodes of the same columns.
We think of nodes in the same columns as different possibilities for a given task,  as different persons performing the same job,  as different machines in the same workstation, or as variations of the same kind of treatment.  

\begin{figure}
\centering
\tikzset{
node_style/.style={circle,draw=black,fill=black!20!, minimum size=1cm},
edge_style/.style={draw=black,  thick, ->},
label/.style={sloped,above}
}
\begin{tikzpicture}[shorten >=1pt,
   auto,
   node distance=2cm and 10cm,
   scale=0.9]
    \node[node_style] (s) at (0,-3) {$s$};
     \foreach \y/\t in {1/1, 2/1, 3/1, 4/1}{
     \node[node_style]  (v\y) at (3, -2*\y+2) {(\y,\t)};
     \draw [edge_style]  (s) edge  (v\y);}
       
     \foreach \y/\t in {1/2, 2/2, 3/2}
     \node[node_style]  (w\y) at (6, -2*\y+1) {(\y,\t)};
     \foreach  \y in {1,2,3,4} 
        \foreach  \w in {1,2,3} 
           \draw [edge_style]  (v\y) edge  (w\w);
           
     \foreach \y/\t in {1/3, 2/3}
     \node[node_style]  (x\y) at (9, -2*\y) {(\y,\t)};
     \foreach  \y in {1,2,3} 
        \foreach  \w in {1,2} 
           \draw [edge_style]  (w\y) edge  (x\w);
           
      \foreach \y/\t in {1/4, 2/4, 3/4, 4/4}
     \node[node_style]  (y\y) at (12, -2*\y+2) {(\y,\t)};
     \foreach  \y in {1,2} 
        \foreach  \w in {1,2,3,4} 
           \draw [edge_style]  (x\y) edge  (y\w);
      
      \node[node_style] (t) at (15,-3) {$t$};     
           \foreach  \w in {1,2,3,4} 
            \draw [edge_style]   (y\w) edge (t);
    \end{tikzpicture}
\caption{An illustration of a DAG with $c=4$ and  $r_{1}=4, r_{2}=3, r_{3}=2$ and $r_{4}=4$. Every node in column $i$ has outgoing edges to every node in column $i+1,~ i=1,\ldots, c-1$.}\label{fig:DAG1}
\end{figure}
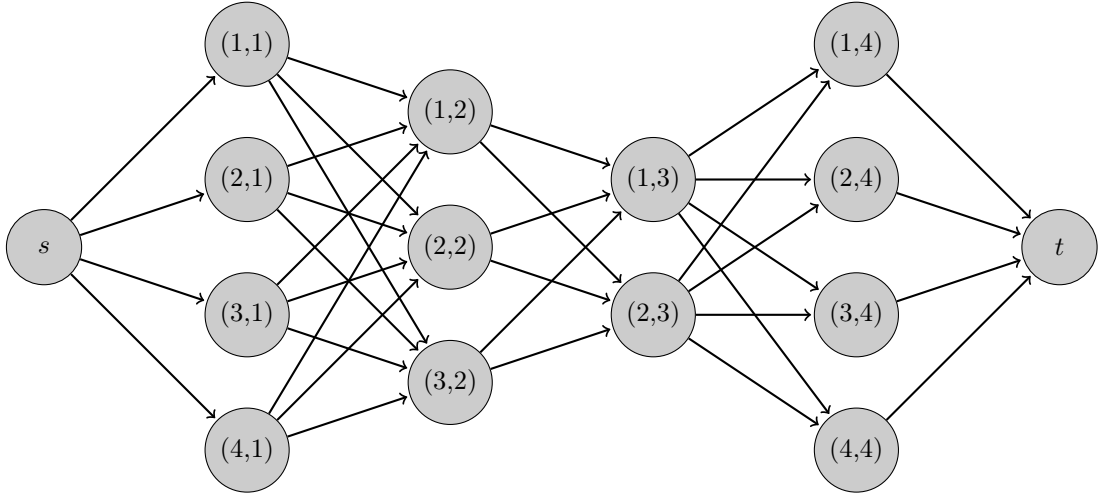

The given data consists of the list of 
paths $\{ \vec{p}_{i} \}_{i = 1,\ldots, n}$ in the DAG and the list of outputs $\{ b(\vec{p}_{i}) \}_{i = 1,\ldots, n}$ for each path.
We consider the {quality matrix} $S$, which is a random matrix of size $r\times c$ with real entries:
$$
S := ( s(i,j) )_{i\in[r],\, j\in [c]},
\qquad
s(i,j) \in \mathbb{R}.
$$
We model the paths with a random vector $\rP := (P_{1}, \ldots, P_{c})$ 
where the components $P_{i}$ are random variables over the set $[r]$. 

Throughout the paper we work under the following standing assumptions:
\begin{ass}\label{ass1} We assume that:
\begin{enumerate}
 \item   all entries of $S$ have finite second moments; 
 \item  the random variables $S\ij, i\in[r],\, j\in [c]$ are (jointly) independent;
 \item the paths $\rP_{1}, \rP_{2}, \ldots$ are chosen independently and according to a Markov chain, see Section \ref{sec:Markov} for more details.\end{enumerate}
\end{ass}
\noindent Note that,  we do not assume the entries of $S$ to be identically distributed or having the same variance. 

Let $\vec{p} = (p_{0},p_{1},\ldots, p_{c}, p_{c+1})$ be a realization of $\rP$, where $p_{0}=s$ and $p_{c+1}=t$ almost surely.
Then, the quality of the construction path $\vec{p}$ is defined as
$$
b(\vec{p}, S):= \sum_{j= 1}^{c} S(p_{j}, j).
$$
In some situation we will abbreviate $b(\vec{p},S)$ and write only $b(\vec{p})$.  
We also can think of $b(\vec{p})$ as the quality (or error) cumulated along the path $\vec{p}$ or as the response variable $Y$ in the language of linear models.
We want to stress out that the model fits the description in (\ref{eq:linearRelationshipDAG}) and that we will continue to use the ``graph-based'' description of the model since it is best suited for our approach and used methods. 

Let us make precise how the randomness enters in our model. We choose a random path $\rP$ and a random matrix $S$. The corresponding probability measure is denoted by $\P$.
 The random choice of $\rP$ and $S$ induces a random variable $b(\rP)=b(\rP,S)$ and allows to generate a sequence of  i.i.d.~random variables $(\rP_{1}, b(\rP_{1})), (\rP_{2}, b(\rP_{2})),\ldots$.

Our goal is to give estimates on the law of $S$ by observing the paths $\rP$ and its cumulated qualities $b(\rP)$. Note that,  $(\rP_{n}, b(\rP_{n}))_{n\in\N}$ is  in general not a sufficient statistic for $S$, i.e., we can not recover the distribution of $S$ by only observing realizations of $(\rP, b(\rP))$, as we see in the following remark.

\begin{remark}\label{rem:not_sufficient}
Let us consider the case $r=1$ and $c=2$. Let $S(1,1) \sim\normal(0,1)$ and  
 $S(1,2) \sim\normal(1,1)$ and define  $\tilde S(1,1):= S(1,1)+1$ and  
 $\tilde S(1,2):= S(1,2)-1$. Then for any given path $\path$ we have that
$\sum_{j= 1}^{2} S(p_{j}, j)=  \sum_{j= 1}^{2} \tilde S(p_{j}, j)$. Hence, the statistic $(\path_{n}, b(\path_{n}))_{n\in \N}$ does not allow us to distinguish between $S$ and $\tilde S$. 
\end{remark}

\begin{ex}[Binary errors]
The matrix $S$ consists of  independent Bernoulli random variables $S(i,j)$. The value $1$ of this Bernoulli may encode a defect and hence $b(\path)$ counts the number of defects of the end product.
\end{ex}

\begin{ex}[Gaussian error]\label{ex:gaussian}
The matrix $S$ consists of  independent Gaussian  random variables $S(i,j)$. The quality or response $b(\path)$ is then distributed as a (random) mixture of Gaussian random variables.
 \end{ex}

Given a sequence of realizations $(\vec{p}_{k})_{k\in[n]}$ of $\rP$, we define the following matrices that are at the core of our analysis:
\begin{equation*}
B^{(n)} := \sum_{k=1}^{n} b(\vec{p}_{k}) V_{\vec{p}_{k}},~
V^{(n)} := \sum_{k=1}^{n} V_{\vec{p}_{k}}, n\geq 1.
\end{equation*}
The value $B^{(n)}(i,j)$ is the sum of all cumulated qualities of paths containing node $\ij$, whereas $V^{(n)}(i,j)$ just counts the number of times node $\ij$ was used.
We define the {sample mean matrix} as the sample mean quality matrix  :
$$
T^{(n)} := (T^{(n)}(i,j))_{i\in[r],\, j\in [c]}, 
\mbox{ where }
T^{(n)}(i,j)
:=
\begin{cases}
\frac{B^{(n)}(i,j)}{V^{(n)}(i,j)},&\quad\text{if } V^{(n)}(i,j) \neq 0 ;\\
\text{0, } &\quad\text{otherwise}.\\
\end{cases}
$$

The corresponding { sample variance matrix} $\Sigma^{(n)}$ is defined by
\begin{equation*}
\Sigma^{(n)}(i,j):= \begin{cases}
\frac{1}{V^{(n)}(i,j)} \sum_{k=1}^{n} \left(b(\vec{p}_{k}) V_{\vec{p}_{k}}(i,j) - T^{(n)}(i,j)\right)^{2},&\quad\text{if } V^{(n)}(i,j) \neq 0 ;\\
\text{0, } &\quad\text{otherwise}.\\
\end{cases}
\end{equation*}

\section{Uniform transition model.}\label{section:UniformTransition}
In this section we review some results from \cite{GuMu:19} in the special case where  the paths $\rP_{1}, \rP_{2}, \ldots$ are chosen independently and uniformly.

\begin{thm}\label{thm:lln}
Let $(i,j), (i', j) \in [r] \times [c]$, then 
$$
T^{(n)}(i,j) - T^{(n)}(i',j)\xrightarrow[n\rightarrow \infty]{a.s.} \E[S(i,j)] - \E[S(i',j)]
$$
and 
$$
\Sigma^{(n)}(i,j) - \Sigma^{(n)}(i',j)\xrightarrow[n\rightarrow \infty]{a.s.} \V[S(i,j)] - \V[S(i',j)].
$$
\end{thm}
\begin{proof}[\textbf{Proof}] We give the proof from \cite{GuMu:19} since it illustrates the idea behind the construction of the estimators. Using twice the law of large numbers and the continuous mapping theorem, we obtain
\begin{eqnarray*}
T^{(n)}(i,j) 
&= &
\frac{1}{|D^{(n)}_{(i,j)}|} 
\sum_{\vec{p} \in D^{(n)}_{(i,j)} }
b(\vec{p}) \cr
&=& \frac{n}{|D^{(n)}_{(i,j)}|} \frac1n
\sum_{\vec{p} \in D^{(n)}_{(i,j)} }
b(\vec{p}) \cr
&\xrightarrow[n\rightarrow \infty]{a.s.}& \frac{1}{\P(\rP[j]=i)} \E[b(\rP); \rP[j]=i'] = \E\left[b(\rP) \, | \, \rP[j]=i\right].
\end{eqnarray*}
In the same way
\begin{equation*}
T^{(n)}(i',j) 
= 
\frac{1}{|D^{(n)}_{(i',j)}|} 
\sum_{\vec{p} \in D^{(n)}_{(i',j)} }
b(\vec{p})
\xrightarrow[n\rightarrow \infty]{a.s.}
\E\left[b(\rP) \, | \, \rP[j]=i'\right].
\end{equation*}
Using the assumption that the paths are chosen uniformly and the definition of $b(\rP)$, we obtain the first part of the theorem from
\begin{eqnarray*}
\E\left[b(\rP)| \, \rP[j]=i\right] - \E\left[b(\rP) \, | \, \rP[j]=i'\right] &=& \frac{\E\left[b(\rP);  \rP[j]=i\right]- \E\left[b(\rP); \rP[j]=i'\right] }{\P(\rP[j]=i)} \cr
&=& \frac{\E\left[S(i,j);  \rP[j]=i\right]- \E\left[S(i',j); \rP[j]=i\right] }{\P(\rP[j]=i)} \cr
&=& \E[S(i,j)] - \E[S(i',j)].
\end{eqnarray*}
For the second part of the theorem, we use the law of large numbers and the continuous mapping theorem to get that
\begin{eqnarray*}
\Sigma^{(n)}(i,j) 
&=&  
\left(\frac{1}{|D_{ij}^{(n)}|} \sum_{\vec{p} \in D_{ij}^{(n)} }  b^2(\vec{p})
\right)
-  \left(
T^{(n)}(i,j)
\right)^{2}
\cr
&\xrightarrow[n\rightarrow \infty]{a.s.}&
\E\left[b(\rP)^{2}| \, \rP[j]=i\right] -\E\left[b(\rP)| \, \rP[j]=i\right]^{2}.
\end{eqnarray*}
Using the definition of  $b(\path)$ we deduce  with elementary calculations that
$$
\E\left[b(\rP)^{2}| \, \rP[j]=i\right] -\E\left[b(\rP)| \, \rP[j]=i\right]^{2}
=
A_j
+
\var{S(i,j)}
$$
where $A_j$ is a quantity that only depend on the column $j$. Applying this identify for $(i,j)$ and $(i',j)$ we obtain that
$$
\Sigma^{(n)}(i,j) - \Sigma^{(n)}(i',j)
\xrightarrow[n \rightarrow \infty]{a.s.}
\var{S(i,j)} - \var{S(i',j)}.
$$
\end{proof}

\section{Markov chain transition model.}\label{sec:Markov}
In this Section, we will treat the more general case where the random path $\rP$ is the path of a  Markov chain. This case corresponds to the situation where the covariates may be dependent, and their correlation follows a Markov structure.

More formally,  we consider the time-inhomogeneous Markov chain on $\{1, \ldots, r\}$ with initial state $s$ and absorbing state $t$.
At the time $0$, we start the Markov chain in the source $s$. 
 The transition kernel for the first step is given by 
 \begin{equation}
 Q^{(1)}(s, i) 
=
\begin{cases}
q^{(1)}_{s,i},    \mbox{ if }\, i \leq r_{1};\\
0,  \mbox{ otherwise.}
\end{cases}
\label{eq:condition1}
\end{equation}
The next $c-1$ steps are defined as follows. For $2\leq k \leq c$:
\begin{equation}
 Q^{(k)}(i,j)
=
\begin{cases}
q^{(k)}_{(i,j)},   \mbox{ if }\, i\in[r_{k-1}], j\in[r_{k}]);\\
0,  \mbox{ otherwise}.
\end{cases}
\label{eq:condition2}
\end{equation}
Finally, the last step is determined by
\begin{equation}
 Q^{(c+1)}(i, t) =
\begin{cases}
1,  \mbox{ if } i\in[r_{c}];
\\
0,  \mbox{ otherwise.}
\end{cases}
\label{eq:condition3}
\end{equation}
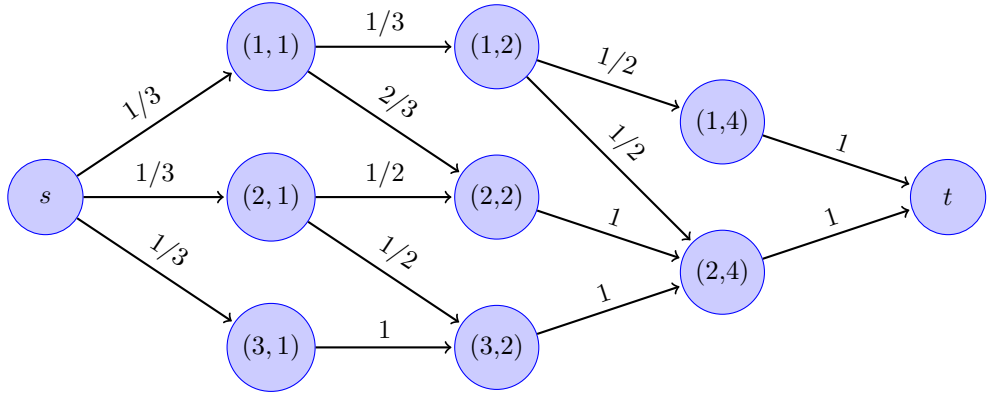
\begin{figure}
\tikzset{
node_style/.style={circle,draw=blue,fill=blue!20!, minimum size=1cm},
edge_style/.style={draw=black,  thick, ->},
label/.style={sloped,above}
}
  \centering
\begin{tikzpicture}[shorten >=1pt,
   auto,
   node distance=2cm and 10cm]
    \node[node_style] (s) at (0,-2) {$s$};
     \foreach \y/\t in {1/1, 2/1, 3/1}{
     \node[node_style]  (v\y) at (3, -2*\y+2) {$(\y,\t)$};
     \draw [edge_style]  (s) -- node [label] {$1/3$}  (v\y) ;}
       
     \foreach \y/\t in {1/2, 2/2, 3/2}
     \node[node_style]  (w\y) at (6, -2*\y+2) {(\y,\t)};
     \foreach  \y in {1} 
        \foreach  \w in {1} 
           \draw [edge_style]  (v\y) -- node [label] {$1/3$}  (w\w);
      \foreach  \y in {1} 
        \foreach  \w in {2} 
           \draw [edge_style]  (v\y) -- node [label] {$2/3$}  (w\w);
      \foreach  \y in {2} 
        \foreach  \w in {2,3} 
           \draw [edge_style]  (v\y) -- node [label] {$1/2$}  (w\w);
       \foreach  \y in {3} 
        \foreach  \w in {3} 
           \draw [edge_style]  (v\y) -- node [label] {$1$} (w\w);

      \foreach \y/\t in {1/4, 2/4}
     \node[node_style]  (y\y) at (9, -2*\y+1) {(\y,\t)};
     \foreach  \y in {1} 
        \foreach  \w in {1,2} 
           \draw [edge_style]  (w\y) -- node [label] {$1/2$}  (y\w);
      \foreach  \y in {2,3} 
        \foreach  \w in {2} 
           \draw [edge_style]  (w\y) -- node [label] {$1$}  (y\w);
      
      \node[node_style] (t) at (12,-2) {$t$};     
           \foreach  \w in {1,2} 
            \draw [edge_style]   (y\w) -- node [label] {$1$} (t);
    \end{tikzpicture}
    \caption{An example of the Markov chain transition model.}
 \end{figure}
The matrices $Q^{(k)}, k\in[c+1]$, are all supposed to be stochastic matrices.
Equation~\eqref{eq:condition1} forces the initial state $s$ to jump into a state of the first column. Equation~\eqref{eq:condition2}
forces a node in column $j$ to jump into a node of column $j+1$. In the case that the
node is in the last column $c$, Equation~\eqref{eq:condition3} forces it to jump into the final absorbing state $t$. 

\begin{remark}\label{rem:unifasMarkov}
The uniform transition model is a special case of the Markov chain transition model. In fact, setting $q^{(1)}_{s,i}:=1/ r_{i},~ i\in[r_{1}]$ and $q^{(k)}_{i,j}:=1/ r_{k},~i\in[r_{k-1}], j\in[r_{k}]$ yields the uniform transition model.
\end{remark}

A Markov chain, as given above, defines a natural probability measure on the paths in the DAG. The measure that chooses a path according to $Q:=(Q^{(1)}, Q^{(2)}, \ldots Q^{(c+1)})$, and the values along this path according to $S$  is  denoted by $\P^{Q}$. We define the support of the measure $Q$ as 
\begin{equation*}
\supp^{Q}   :=  \left\{   \path : \P^{Q}(\rP = \path) > 0 \right\}.\\
\end{equation*}
A node is reachable if there exists a path with positive probability through this node. We enumerate the reachable nodes in each column starting from $1$. The set 
\begin{equation*}
\supp^{Q}_{i,j}:= \{  \path\in \supp^{Q} :  \path[j] = i \}
\end{equation*}
is the set of all paths of positive probability passing through $\ij$. We consider the (multi-)set $D^{(n)} = (\rP_1, \ldots, \rP_n )$ of the first $n$ paths and the multi-set  $D^{(n)}_{\ij}$ containing only  those paths passing through $\ij$.

\begin{remark}\label{rem:MChainEqualDistribution}
Let $(i, k)$ and $(i', k)$ be two $Q$-reachable nodes in the same column.
A consequence of the proof of Theorem \ref{thm:lln} is that if for all $\ell\in[r]$
\begin{eqnarray*}
q^{(k)}(\ell,i) &=& q^{(k)}(\ell, i'), \cr
 q^{(k+1)}(i, \ell) &=& q^{(k+1)}(i', \ell),
\end{eqnarray*}
then
\begin{equation*}
T^{(n)}\ij - T^{(n)} (i',j)\xrightarrow[n\rightarrow \infty]{\P^{Q}-a.s.}  \E[S(i,j)]- \E[S(i',j)].
\end{equation*}
and 
\begin{equation*}
	\Sigma^{(n)}(i,j) - \Sigma^{(n)}(i',j)\xrightarrow[n\rightarrow \infty]{\P^{Q}-a.s.} \V[S(i,j)] - \V[S(i',j)].
\end{equation*}
\end{remark}
 Let $Q:=(Q^{(1)}, Q^{(2)}, \ldots, Q^{(c+1)})$ et $\widetilde Q:=(\widetilde Q^{(1)}, \widetilde Q^{(2)}, \ldots \widetilde Q^{(c+1)})$ be two sequences of  transition matrices. We say that $Q$ and $\widetilde Q$ are { equivalent} (as measures) if 
 \begin{equation*}
Q^{(k)}(i,j) >0  \Longleftrightarrow \widetilde Q^{(k)}(i,j)  >0,
\end{equation*}
for all meaningful choices of $i,j$ and $k$. In particular, $Q$ and $ \widetilde Q$ are equivalent if  $P^{Q}=P^{\widetilde Q}$. Moreover, the induced measures on the set of paths are equivalent and there exists a discrete Radon-Nykodym derivative that allows a change of measures.

\begin{lemma}\label{lemma:MeasureChange}
Let $(i, j)$ be a node and let $Q,  \widetilde Q$ be two equivalent sequences of transition matrices and let $f: P^{Q}\times \R^{r\times c} \rightarrow \R$ be a measurable function, then
$$
\E^{Q}\left[f(\rP, S) \frac{\P^{\widetilde{Q}}\left(\rP'=\rP | \rP'[j]=i\right)}{\P^{Q}\left(\rP'=\rP | \rP'[j]=i\right)} \middle| \rP[j]=i\right]
=
\E^{\widetilde{Q}}[f(\rP,S) | \rP[j]=i]
$$where $\rP'$ is an independent copy of $\rP$.
\end{lemma}
\begin{proof}[\textbf{Proof}]
The proof follows from the theorem of total probability, the independence of $S$ and $\rP$, and a change of measure:
\begin{eqnarray*}
\E^{\widetilde{Q} }[f(\rP, S) | \rP[j]=i] 
&=& \sum_{{\path:\path[j]=i}} \E^{\widetilde{Q} }\left[{f(\path, S)}  \right]  \P^{\widetilde{Q} }\left(\rP=\path | \rP[j]=i\right) \cr
&=&\sum_{ {\path \in \mbox{ \footnotesize Supp}^Q_{i,j}} }  \E^{Q } \left[{f(\path,S)} \right]  \P^{{Q}}\left(\rP=\path | \rP[j]=i\right) \frac{\P^{\widetilde{Q} }\left(\rP=\path | \rP[j]=i\right)}{\P^{Q }\left(\rP=\path | \rP[j]=i\right)} \cr
&=&\sum_{{\path \in \mbox{ \footnotesize Supp}^Q_{i,j}}} \E^{Q }\left[{f(\path,S)}\frac{\P^{\widetilde{Q} }\left(\rP=\path | \rP[j]=i\right)}{\P^{Q }\left(\rP=\path | \rP[j]=i\right)}  \right]  \P^{Q }\left(\rP=\path | \rP[j]=i\right)  \cr
&=& \E^{Q }\left[{f(\rP,S)} \frac{\P^{\widetilde{Q} }\left(\rP'=\rP | \rP'[j]=i\right)}{\P^{Q }\left(\rP'=\rP | \rP'[j]=i\right)} \middle| \rP[j]=i\right].
\end{eqnarray*}
\end{proof}
We can now  define the estimators for the means and variances.
\begin{defn}
We define 
\begin{eqnarray*}
T^{(n)}_{Q, \widetilde Q}(i,j) := \frac{1}{|D^{(n)}_{(i,j)}|} \sum_{\path \in D^{(n)}_{\ij}} \left( b(\path)   \frac{\P^{\widetilde Q}\left(\rP=\path \middle| \rP[j]=i \right)}{\P^{Q}\left(\rP=\path \middle| \rP[j]=i \right)}\right)
\end{eqnarray*}
and 
\begin{eqnarray*}
\Sigma^{(n)}_{Q,\widetilde Q}(i,j) &:=&
   \left(\frac{1}{|D^{(n)}_{(i,j)}|}
   \sum_{_{\path \in D^{(n)}_{\ij}}}
   b^2(\path)   \frac{\P^{\widetilde Q}\left(\rP=\path \middle| \rP[j]=i \right)}{\P^{Q}\left(\rP=\path \middle| \rP[j]=i \right)}
   \right)
   -
   \left(
	   T^{(n)}_{Q, \widetilde Q}(i,j).
   \right)^2
\end{eqnarray*}
\end{defn}

\begin{thm}\label{thm:markovlln}
  Let $(i, j)$ be a node and let $Q, \widetilde Q$ be two equivalent sequences of transition matrices.
  Then,
  \[
  T^{(n)}_{Q, \widetilde Q}(i,j) 
    \xrightarrow[n\rightarrow \infty]{\P^{Q}-a.s.}
    \E^{\widetilde Q}\left[b(\rP)  \middle| \rP[j]=i \right],
  \]
   \[
   \Sigma^{(n)}_{Q,\widetilde Q}(i,j) \xrightarrow[n\rightarrow \infty]{\P^{Q}-a.s.}
   \V^{\widetilde Q}[b^2(\rP)  | \rP[j]=i].
 \]
\end{thm}

\begin{proof}[\textbf{Proof}]
The first observation is that
\begin{equation*}
T^{(n)}_{Q, \widetilde Q}(i,j) = \frac{n}{|D^{(n)}_{(i,j)}|} \cdot \frac1n \sum_{k=1}^{n} \left( \ind\{\path_{k}[j]=i\} b(\path_{k})   \frac{\P^{\widetilde Q}\left(\rP=\path_{k} \middle| \rP[j]=i \right)}{\P^{Q}\left(\rP=\path_{k} \middle| \rP[j]=i \right)}\right).
\end{equation*}
Using twice the law of large numbers and the continuous mapping theorem we obtain
\begin{eqnarray*}
T^{(n)}_{Q, \widetilde Q}(i,j) & \xrightarrow[n\rightarrow \infty]{\P^{Q}-a.s.} & \P^{Q}( \rP[j]=i)^{-1} \E^{Q} \left[ \ind\{\rP[j]=i\} b(\rP)   \frac{\P^{\widetilde Q}\left(\rP'=\rP \middle| \rP'[j]=i \right)}{\P^{Q}\left(\rP'=\rP \middle| \rP'[j]=i \right)}    \right] \cr
&=& \E^{Q} \left[  b(\rP)   \frac{\P^{\widetilde Q}\left(\rP'=\rP \middle| \rP'[j]=i \right)}{\P^{Q}\left(\rP'=\rP \middle| \rP'[j]=i \right)} \middle| \rP[j]=i   \right],
\end{eqnarray*}
where $\rP'$ is an independent copy of  $\rP$.
Now, using Lemma~\ref{lemma:MeasureChange} with $f=b$ and recalling that $b(\rP)=b(\rP,S)$, we get that:
$$
\E^{ Q}\left[b(\rP) \frac{\P^{\widetilde Q}\left(\rP'=\rP | \rP'[j]=i\right)}{\P^{ Q}\left(\rP'=\rP | \rP'[j]=i\right)} \middle| \rP[j]=i\right]
=
\E^{\widetilde Q}[b(\rP) | \rP[j]=i].
$$
For the second part of the theorem, we note that
\begin{align*}
\Sigma^{(n)}_{Q,\widetilde Q}(i,j)
&=
\left(\frac{1}{|D^{(n)}_{(i,j)}|}
\sum_{_{\path \in D^{(n)}_{\ij}}}
b^2(\path)   \frac{\P^{\widetilde Q}\left(\rP=\path \middle| \rP[j]=i \right)}{\P^{Q}\left(\rP=\path \middle| \rP[j]=i \right)}
\right) - \left( T^{(n)}_{Q, \widetilde Q}(i,j) \right)^2
\\& \xrightarrow[n\rightarrow \infty]{\P^{Q} a.s.}
\E^{Q}\left[b^2(\rP)   \frac{\P^{\widetilde Q}\left(\rP'=\rP \middle| \rP'[j]=i \right)}{\P^{Q}\left(\rP'=\rP \middle| \rP'[j]=i \right)}\middle| \rP[j]=i  \right]
-
\left( \E^{\widetilde Q } \left[ b(\rP) \bigg|\rP[j]=i \right] \right)^2.
\end{align*}
Using again Lemma~\ref{lemma:MeasureChange} but this time with $f= b^2$, we get that
\begin{equation*}
\E^{Q}\left[b^2(\rP)   \frac{\P^{\widetilde Q}\left(\rP'=\rP \middle| \rP'[j]=i \right)}{\P^{Q}\left(\rP'=\rP \middle| \rP'[j]=i \right)}
\middle|
\rP[j]=i
\right] =
\E^{\widetilde Q} \left[ b^2(\rP)\middle|\rP[j]=i\right],
\end{equation*}
therefore
$$
\Sigma_{Q,\widetilde Q}(i,j)
\xrightarrow[n\rightarrow \infty]{\P^{Q} a.s.}
\E^{\widetilde Q}
\left[
b^2(\rP)\bigg|\rP[j]=i\right]
-
\left(
\E^{\widetilde Q}
\left[
b(\rP)
\bigg|\rP[j]=i
\right]
\right)^2
=
\V^{\widetilde Q}[b^2(\rP)  | \rP[j]=i].
$$
\end{proof}

We obtain the following consequence of the Theorems \ref{thm:lln} and \ref{thm:markovlln}.
\begin{corollary}\label{cor:markovuniform}
Let $\widetilde Q=U$ be the transition matrix corresponding to the uniform transitions. If $Q$ and $\widetilde Q$ are equivalent transition matrices then
\begin{equation*}
T^{(n)}_{Q, U}(i,j) - T^{(n)}_{Q, U}(i',j)   \xrightarrow[n\rightarrow \infty]{\P^{Q}-a.s.}  \E[S(i,j)]- \E[S(i',j)],
\end{equation*}
and 
$$
	\Sigma_{Q,U}^{(n)}(i,j) - \Sigma_{Q,U}^{(n)}(i',j)\xrightarrow[n\rightarrow \infty]{\P^{Q}-a.s.} \V[S(i,j)] - \V[S(i',j)].
$$

\end{corollary}

\begin{remark} In the case $\widetilde Q=U$ the estimators can be written in a more explicit form. For example,
\begin{equation*}
T^{(n)}_{Q, U}(i,j) = \frac{ \prod_{\ell\neq j} (r_{\ell})^{-1} }{|D^{(n)}_{(i,j)}|}  \sum_{\path \in D^{(n)}_{\ij}} \left( b(\path) \frac{ \P^{Q}\left( \rP[j]=i \right)}{\P^{Q}\left(\rP=\path \right)}\right).
\end{equation*}
\end{remark}

\section{Asymptotic distribution of the estimators}\label{sec:asympt}

\subsection{Known $Q$ and $\widetilde Q$}

In the case where  both distributions, $Q$ and $\widetilde Q$, are known, we deduce the asymptotic distribution of the estimators $T^{(n)}_{Q, \widetilde Q}(i,j) $ and $\Sigma^{(n)}_{Q,\widetilde Q}(i,j)$.
\begin{thm}\label{thm:Qknown}
We have that
\begin{equation*}
	\sqrt{|D^{(n)}_{(i,j)}|}
	(T_{Q,\widetilde Q}^{(n)}(i,j) - \mu_{i,j})
	\xrightarrow[n \rightarrow \infty]{{\P^{Q}-}distr.}
	\mathcal{N}(0, \sigma_{i,j}^{2})
\end{equation*}
where
\begin{equation*}
	\mu_{i,j} 
	= 
	\E^{\widetilde Q }
	[ b(\rP) | \rP[j]=i ] 
	\quad \mbox{ and } \quad
	\sigma^2_{i,j}
	=
	\E^{{\widetilde Q}}[b(\rP)^2|\rP[j]=i]
	-
	\E^{{\widetilde Q}}[b(\rP)|\rP[j]=i]^2.
\end{equation*}
If we assume that the entries of $S$ have finite forth moments, we have that
\begin{equation*}
\sqrt{|D^{(n)}_{(i,j)}|} \left(\Sigma^{(n)}_{Q,\widetilde Q}(i,j) -\sigma^2_{i,j}\right)
\xrightarrow[n\rightarrow \infty]{\P^{Q}-distr.}
\mathcal{N}(0, {(1,-1)\Sigma_{i,j}(1,-1)^T} ),
\end{equation*}
 where
{\begin{equation*}
\Sigma_{i,j} =
 \begin{bmatrix}
\mathbb{V}^{Q}\left[ b^2(\rP)C_{i,j}(\rP) \big| \rP[j]=i \right] & 2 \mu_{i,j} \mathrm{Cov}^Q\left[ b^2(\rP)C_{i,j}(\rP), b(\rP)C_{i,j}(\rP) \big| \rP[j]=i\right] \\
2 \mu_{i,j}\mathrm{Cov}^Q\left[ b^2(\rP)C_{i,j}(\rP), b(\rP)C_{i,j}(\rP) \big| \rP[j]=i\right]  &
4 \mu_{i,j}^{2}\mathbb{V}^{Q}\left[b(\rP) C_{i,j} (\rP) \big| \rP[j]=i \right] \\
\end{bmatrix}
\end{equation*}}
and $C_{i,j}(\path)= \frac{\P^{\widetilde Q}\left(\rP=\path \middle| \rP[j]=i \right)}{\P^{Q}\left(\rP=\path \middle| \rP[j]=i \right)}$.
\end{thm} 

\begin{proof}[\textbf{Proof}]
Let us observe that
\begin{align*}\label{eq:EstimatorT'}
T^{(n)}_{Q, \widetilde Q}(i,j) 
&= 
\frac{1}{|D^{(n)}_{(i,j)}|} \sum_{\path \in D^{(n)}_{\ij}} b(\path) \frac{\P^{\widetilde Q}\left(\rP=\path \middle| \rP[j]=i \right)}{\P^{ Q}\left(\rP=\path \middle| \rP[j]=i \right)}
\\&=
{\frac{1}{|D^{(n)}_{(i,j)}|} \sum_{k=1}^{n} b(\rP_k) \frac{\P^{\widetilde Q}\left(\rP=\rP_k \middle| \rP[j]=i \right)}{\P^{ Q}\left(\rP=\rP_k \middle| \rP[j]=i \right)}\ind\left\{\rP[j]=i\right\}.}
\end{align*}
The last sum can be interpreted as a sum of i.i.d.~random variables appearing in an acceptance-rejection sampling. More precisely, we start with 
$\left(\rP_1,b(\rP_1) \frac{\P^{\widetilde Q}\left(\rP=\rP_1 \middle| \rP[j]=i \right)}{\P^{ Q}\left(\rP=\rP_1 \middle| \rP[j]=i \right)}\right)$ 
and $k=1$; if 
$\rP_k[j]=i$ then we set 
\begin{equation}
Y:=b(\rP_k) \frac{\P^{\widetilde Q}\left(\rP=\rP_k \middle| \rP[j]=i \right)}{\P^{ Q}\left(\rP=\rP_k \middle| \rP[j]=i \right)}
\end{equation}
and stop, otherwise we increase $k$ and repeat until $\rP_K[j]=i$ for the first K. Now, letting 
\begin{equation}
f(\rP_k) := b(\rP_k) \frac{\P^{\widetilde Q}\left(\rP=\rP_k \middle| \rP[j]=i \right)}{\P^{ Q}\left(\rP=\rP_k \middle| \rP[j]=i \right)}
\end{equation}
 we obtain for $y \in \R$ that
\begin{equation*}
\P^{ Q}[Y\leq y]
=
\sum_{k=1}^{\infty}\P^{Q}[f(\rP_k)\leq y | K=k] \P^{Q}[K=k]
=
\P^{ Q}[f(\rP_1) | \rP_1[j]=i];
\end{equation*}
this means that the distribution of $Y$ equals the distribution of $f(\rP_1)$ conditioned on {$\rP_{1}[j]=i.$} Iterating this acceptance-rejection method, we see that $|D_{i,j}^{(n)}|$ describes the number of acceptances among $(\rP_k, f(\rP_k))$ for $1\leq k \leq n$. Therefore, the estimator $T_{Q,\widetilde Q}^{(n)}$ has the same distribution as
$$
	\frac{1}{|D_{i,j}^{(n)}|}
	\sum_{k=1}^{|D_{i,j}^{(n)}|} 
	Y_k,
$$
where $Y_k$, $k \in \N$ is a sequence of i.i.d.\ random variables distributed as $f(\rP_1)$ and conditioned on $\rP_1[j]=i$. Finally, Anscombe's Theorem \cite{Gu:09} implies that
\begin{equation}
	\sqrt{|D^{(n)}_{(i,j)}|}
	(T_{Q,\widetilde Q}^{(n)}(i,j) - \mu_{i,j})
	\xrightarrow[n \rightarrow \infty]{{\P^{Q}-}distr.}
	\mathcal{N}(0, \sigma_{i,j})
\end{equation}
where
\begin{equation*}
	\mu_{i,j} 
	= 
	\E^{\widetilde Q }
	[ b(\rP_1) | \rP_1[j]=i ] 
	\quad \mbox{ and } \quad
	\sigma^2_{i,j} 
	=
	\E^{\widetilde Q}[f(\rP)^2|\rP[j]=i]
	-
	\E^{\widetilde Q}[f(\rP)|\rP[j]=i]^2.
\end{equation*}

For the variance, we assume that the entries of $S$ have finite fourth moments.
The estimator for the variance has the following form:
\begin{equation}
\Sigma^{(n)}_{Q,\widetilde Q}(i,j)
=
\left(\frac{1}{|D^{(n)}_{(i,j)}|}
\sum_{_{\path \in D^{(n)}_{\ij}}}
b^2(\path)   \frac{\P^{\widetilde Q}\left(\rP=\path \middle| \rP[j]=i \right)}{\P^{Q}\left(\rP=\path \middle| \rP[j]=i \right)}
\right) - \left( T^{(n)}_{Q, \widetilde Q}(i,j) \right)^2;
\end{equation}
We consider the following i.i.d.\ vectors $Y_k=(Y_{1,k},Y_{2,k})$ 
where $Y_{k}$ is distributed  as 
$$
\left(b^2(\rP)   \frac{\P^{\widetilde Q}\left(\rP'=\rP \middle| \rP'[j]=i \right)}{\P^{Q}\left(\rP'=\rP \middle| \rP'[j]=i \right)}, b(\rP)   \frac{\P^{\widetilde Q}\left(\rP'=\rP \middle| \rP'[j]=i \right)}{\P^{Q}\left(\rP'=\rP \middle| \rP'[j]=i \right)} \right)
$$  under $\P^{Q}$ conditioned on $\rP'[j]=i$. 
We write  $C_{i,j}(\path):= \frac{\P^{\widetilde Q}\left(\rP=\path \middle| \rP[j]=i \right)}{\P^{Q}\left(\rP=\path \middle| \rP[j]=i \right)}$ and, due to Lemma \ref{lemma:MeasureChange}, for $k\in\{1,2\}$
 \begin{equation}\label{eq:MuQToMuQTilde}
\mu_{i,j}^{(k)}:=\E^{{ Q}}[b(\rP)^k C_{i,j}(\rP) |\rP[j]=i] = \E^{\tilde{ Q}}[b(\rP)^k |\rP[j]=i],
\end{equation}
and covariance matrix 
{\begin{equation}
	\Sigma'_{i,j} = 
	\begin{bmatrix}
	\mathbb{V}^{Q}\left[ b^2(\rP)C_{i,j}(\rP) \big| \rP[j]=i \right] &  \mathrm{Cov}^Q\left[ b^2(\rP)C_{i,j}(\rP), b(\rP)C_{i,j}(\rP) \big| \rP[j]=i\right] \\
	\mathrm{Cov}^Q\left[ b^2(\rP)C_{i,j}(\rP), b(\rP)C_{i,j}(\rP) \big| \rP[j]=i\right]  &
	\mathbb{V}^{Q}\left[b(\rP) C_{i,j} (\rP) \big| \rP[j]=i \right] \\
\end{bmatrix}.
\end{equation}}

We conclude from Theorem~\ref{thm:MultiDimensionalAnscombe} that
\begin{equation}
\sqrt{|D^{(n)}_{(i,j)}|}
\left( 
\frac{1}{|D^{(n)}_{(i,j)}|}\sum_{_{\path \in D^{(n)}_{\ij}}}
\left(b^2(\path)C_{i,j}(\path) -\mu_{i,j}^{(2)}\right),
\frac{1}{|D^{(n)}_{(i,j)}|}\sum_{_{\path \in D^{(n)}_{\ij}}}
\left(b(\path)C_{i,j}(\path) -\mu_{i,j}^{(1)}\right)
\right)
\xrightarrow[n\rightarrow \infty]{\P^{Q}-distr.}
\mathcal{N}(0, \Sigma'_{i,j} ).
\end{equation}
  
Plugging in the definition of $T^{(n)}_{Q, \widetilde Q}(i,j)$ and using Anscombe's delta method (see Theorem~\ref{thm:AscombesDeltaMethod}) with the function $h(x,y)=\left(x, \left(y+\mu^{(1)}_{i,j}\right)^2 - \mu^{(1)}_{i,j} \right)$ we obtain
\begin{align}\label{eq:DeltaMethodApplication}
\sqrt{|D^{(n)}_{(i,j)}|}&
\left(\frac{1}{|D^{(n)}_{(i,j)}|} 
\sum_{_{\path \in D^{(n)}_{\ij}}}\left(
b^2(\path)   C_{i,j}(\path) -\mu_{i,j}^{(2)}\right),    \left( T^{(n)}_{Q, \widetilde Q}(i,j) \right)^2 -{\mu_{i,j}^{(1)}}^{2}\right) \cr
&\hspace{2cm} \xrightarrow[n\rightarrow \infty]{\P^{Q}-distr.}
\mathcal{N} \left(0, \left(\nabla h(\vec{0})\right)^T \cdot \Sigma'  \cdot  \nabla h(\vec{0})\right),
\end{align}
where 
\begin{equation}
	\nabla h(\vec{0}) = 
	J_h(\vec{0})=
	\begin{bmatrix}
		1 & 0 \\
		0 & 2\mu^{(1)}_{i,j} \\ 
	\end{bmatrix}.
\end{equation}
Noting that $\mu_{i,j}^{(1)} = \mu_{i,j}$, by Equation~(\ref{eq:MuQToMuQTilde}), the latter results in 
\begin{equation}
	\begin{split}
	\Sigma_{i,j} := \left(\nabla h(\vec{0})\right)^T \cdot \Sigma' \cdot  \nabla h(\vec{0}) =
	\begin{bmatrix}
		\Sigma'_{i,j}(1,1) & 2 \mu_{i,j} \Sigma'_{i,j}(1,2) \\
		 2 \mu_{i,j}\Sigma'_{i,j}(2,1)  &
		4 \mu_{i,j}^{2} \Sigma'_{i,j}(2,2) \\
	\end{bmatrix}.
	\end{split}
\end{equation}
Hence, by the continuous mapping theorem with the function $f(x,y)= x-y$ we get that
$$
\sqrt{|D^{(n)}_{(i,j)}|} \left(\Sigma_{Q,\widetilde Q}(i,j) - \sigma^2_{i,j}\right)
\xrightarrow[n\rightarrow \infty]{\P^{Q}-distr.}
\mathcal{N}(0, (1,-1) \cdot \Sigma'_{i,j} \cdot (1,-1)^T ),
$$
which concludes the proof. 
\end{proof}

\subsection{Unknown $Q$ and known $\widetilde Q$}
In the case where $Q$ is unknown, we have to replace the probabilities $\P^{Q}\left(\rP=\path  \middle| \rP[j]=i \right)$ by estimators. Define 
\begin{equation}\label{eq:EstimatorC(p,n)}
\widehat{C}_{i,j}(\path,n)
:=
\P^{\widetilde Q}\left(\rP=\path \middle| \rP[j]=i \right)
\left(\frac{\sum_{k=1}^{n} \ind\{\path_k = \path\} }{\sum_{k=1}^{n} \ind\{\path_k[j] = i \} }\right)^{-1}
\end{equation}
which is a random variable that converges $\P^{Q}$-a.s.\ to  
$C_{i,j}(\path)=\frac{\P^{\widetilde Q}\left(\rP=\path \middle| \rP[j]=i \right)}{\P^{Q}\left(\rP=\path \middle| \rP[j]=i \right)}$. Therefore, we propose  the following variation of $T^{(n)}_{Q, \widetilde Q}(i,j)$:  
\begin{equation}\label{eq:EstimatorT'}
\widehat{T}^{(n)}_{Q, \widetilde Q}(i,j) 
:= 
\frac{1}{|D^{(n)}_{(i,j)}|} \sum_{\path \in D^{(n)}_{\ij}} b(\path)\widehat{C}_{i,j}(\path, n).
\end{equation}
The estimator  $\Sigma^{(n)}_{Q, \widetilde Q}(i,j)$ for the variance becomes:
$$
\widehat{\Sigma}_{Q,\widetilde Q}(i,j) 
=
\left(\frac{1}{|D^{(n)}_{(i,j)}|}
\sum_{_{\path \in D^{(n)}_{\ij}}}
b^2(\path) \widehat{C}(\path,n)
\right)
-
\left(
\widehat{T}^{(n)}_{Q, \widetilde Q}(i,j)
\right)^2.
$$

	\begin{lemma}\label{lemma:ThatSigmahatConvergence}
		Let $Q,\widehat{Q}$ be two equivalent transition matrices then for all $i,j$, we have that
		\begin{equation*}
			\begin{split}
				&\widehat{T}^{(n)}_{Q, \widetilde Q}(i,j) - \mu_{i,j} \xrightarrow[n\rightarrow \infty]{\P^{Q}-a.s } 0;
				\\& \widehat{\Sigma}_{Q,\widetilde Q}(i,j)
				- \sigma^2_{i,j}
				\xrightarrow[n\rightarrow \infty]{\P^{Q}-a.s } 0,
			\end{split}
		\end{equation*}
		where 
		$ \mu_{i,j} 
		= 
		\E^{\widetilde Q }
		[ b(\rP_1) | \rP_1[j]=i ] 
		$
		and
		$
		\sigma^2_{i,j} 
		=
		\E^{{\widetilde Q}}[b(\rP)^2|\rP[j]=i]
		-
		\E^{{\widetilde Q}}[b(\rP)|\rP[j]=i]^2.
		$
	\end{lemma}
	\begin{proof}[\textbf{Proof}]
		Let $q_1,q_2,\ldots,q_m$ be the elements of all pairwise different paths such that $q[j]=i$ and let $D^{(n)}_{\ij} = \{ p_1,p_2,\ldots, p_n \}$, then
		\begin{equation}\label{eq:EquivalentHatT}
		\begin{split}
		\widehat{T}^{(n)}_{Q, \widetilde Q}(i,j) 
		&= 
		\frac{1}{|D^{(n)}_{(i,j)}|} \sum_{\path \in D^{(n)}_{\ij}} b(\path)\widehat{C}_{i,j}(\path, n)
		\\&=
		\frac{n}{D^{(n)}_{\ij}} 
		\sum_{l=1}^{m}
		\widehat{C}_{i,j}(q_l,n) 
		\left( \frac{1}{n} \sum_{k=1}^{n} \ind\{ p_k=q_l \} b(p_k,S_k) \right).
		\end{split}
		\end{equation}
By the strong law of large numbers, we get that
		\begin{equation}\label{eq:almostSureSplit}
			\begin{split}
				& \frac{n}{D^{(n)}_{\ij}} \xrightarrow[n\rightarrow \infty]{\P^{Q}-a.s } \left(\P^{Q}[\rP[j]=i]\right)^{-1};
				\\& \widehat{C}_{i,j}(q_l,n)
				\xrightarrow[n\rightarrow \infty]{\P^{Q}-a.s } C_{i,j}(q_l);
				\\& \frac{1}{n} \sum_{k=1}^{n} \ind\{ p_k=q_l \} b(p_k,S_k)
				\xrightarrow[n\rightarrow \infty]{\P^{Q}-a.s } \E^{Q}[\ind\{ \rP=q_l \} b(\rP,S) ].
			\end{split}
		\end{equation}
Since $m$ is a finite constant, we obtain, using the continuous mapping theorem and Lemma~\ref{lemma:MeasureChange} in the last equality, that
		\begin{equation}\label{eq:AlmostSureConvergenceTn}
			\begin{split}
				\widehat{T}^{(n)}_{Q, \widetilde Q}(i,j) 
				&\xrightarrow[n\rightarrow \infty]{\P^{Q}-a.s }
				\left(\P^{Q}[\rP[j]=i]\right)^{-1}
				\sum_{l=1}^{m}
				C_{i,j}(q_l) 
				\E^{Q}[\ind\{ \rP=q_l \} b(\rP,S) ]
				\\&=
				\left(\P^{Q}[\rP[j]=i]\right)^{-1}
				\sum_{l=1}^{m} 
				\E^{Q}[\ind\{ \rP=q_l \} b(\rP,S) C_{i,j}(\rP) ]
				\\&=
				\left(\P^{Q}[\rP[j]=i]\right)^{-1}
				\E^{Q}[\ind\{ \rP[j]=i \} b(\rP,S) C_{i,j}(\rP) ]
				\\&=
				\E^{Q}[ b(\rP,S) C_{i,j}(\rP) | \rP[j] = i ]
				\\&=
				\E^{\widetilde{Q}}[ b(\rP,S) | \rP[j] = i ]
				= \mu_{i,j}.
			\end{split}
		\end{equation}
For the second part of the statement, we obtain similarly, that
		\begin{equation}\label{eq:AlmostSureTnSquare}
			\left( \widehat{T}^{(n)}_{Q, \widetilde Q}(i,j)\right)^2
			\xrightarrow[n\rightarrow \infty]{\P^{Q}-a.s }
			\E^{\widetilde{Q}}[ b(\rP,S) | \rP[j] = i ]^2.
		\end{equation}
		Furthermore, 
		\begin{equation}
			\begin{split}
				\frac{1}{|D^{(n)}_{(i,j)}|}
				\sum_{_{\path \in D^{(n)}_{\ij}}}
				b^2(\path) \widehat{C}(\path,n)
				=
				\frac{n}{D^{(n)}_{\ij}} 
				\sum_{l=1}^{m}
				\widehat{C}(q_l,n) 
				\left( \frac{1}{n} \sum_{k=1}^{n} \ind\{ p_k=q_l \} b^2(p_k,S_k) \right)
			\end{split}
		\end{equation}
		In the same fashion as in (\ref{eq:almostSureSplit}) and (\ref{eq:AlmostSureConvergenceTn}) we get that
		\begin{equation}\label{eq:AlmostSurebSquare}
			\frac{1}{|D^{(n)}_{(i,j)}|}
			\sum_{_{\path \in D^{(n)}_{\ij}}}
			b^2(\path) \widehat{C}(\path,n)
			\xrightarrow[n\rightarrow \infty]{\P^{Q}-a.s }
			\E^{\widetilde{Q}}[ b^2(\rP,S) | \rP[j] = i ].
		\end{equation}
Hence, using (\ref{eq:AlmostSureTnSquare}) and (\ref{eq:AlmostSurebSquare}) we conclude that
		\begin{equation}
			\begin{split}
			\widehat{\Sigma}_{Q,\widetilde Q}(i,j) 
			=
			\left(\frac{1}{|D^{(n)}_{(i,j)}|}
			\sum_{_{\path \in D^{(n)}_{\ij}}}
			b^2(\path) \widehat{C}(\path,n)
			\right)
			-
			\left(
			\widehat{T}^{(n)}_{Q, \widetilde Q}(i,j)
			\right)^2
			\xrightarrow[n\rightarrow \infty]{\P^{Q}-a.s }
			\sigma^2_{i,j}.
			\end{split} 
		\end{equation}
	\end{proof}

\begin{thm}\label{thm:asymptUnknownQ}   
We have that
\begin{equation*}
	\sqrt{|D^{(n)}_{(i,j)}|}
	(\widehat{T}_{Q,\widetilde Q}^{(n)}(i,j) - \mu_{i,j})
	\xrightarrow[n \rightarrow \infty]{{\P^{Q}-}distr.}
	\mathcal{N}(0, {\vec{C} \Sigma_{i,j}  \vec{C}^T} )
\end{equation*}
where $\mu_{i,j} = \E^{\widetilde Q }[ b(\rP) | \rP[j]=i ]$ and 
{\begin{equation*}
	\Sigma_{i,j}(\ell,\ell')=
	\begin{cases}
	\P^Q\left[ \rP=\vec{q}_\ell \big| \rP[j]=i \right]
	\V^Q\left[ b\left( \vec{q}_\ell , S_{\rP}\right) \big| \rP[j] = i \right], & \ell=\ell';\\
	0, & \mbox{otherwise,}
	\end{cases}
\end{equation*}}
for $\ell,\ell' \in \{1,\ldots, m \}$ and $\vec{C} =  \left(  C_{i,j}(\vec{q}_1), \ldots,  C_{i,j}(\vec{q}_m) \right)$ and $C_{i,j}(\path)= \frac{\P^{\widetilde Q}\left(\rP=\path \middle| \rP[j]=i \right)}{\P^{Q}\left(\rP=\path \middle| \rP[j]=i \right)}$. If we assume that the entries of $S$ have finite forth moments, we have that
\begin{equation*}
\sqrt{|D^{(n)}_{(i,j)}|} \left(\widehat{\Sigma}^{(n)}_{Q,\widetilde Q}(i,j) -\sigma^2_{i,j}\right)
\xrightarrow[n\rightarrow \infty]{\P^{Q}-distr.}
\mathcal{N}(0, {\vec{C}' \Sigma_{i,j}' \vec{C}'^T }),
\end{equation*}
where $\sigma^2_{i,j} 
=
\E^{\widetilde Q}[b(\rP)^2|\rP[j]=i]
-
\E^{\widetilde Q}[b(\rP)|\rP[j]=i]^2$ and 
{\begin{equation}
	\Sigma'_{i,j}(\ell,\ell')=
	\begin{cases}
	\P^Q\left[ \rP = \vec{q}_\ell \big| \rP[j]=i \right] \V^Q\left[ b(\vec{q}_\ell) \big| \rP[j]=i \right], & \ell=\ell' \in \{1,\ldots,m\}; \cr
	\P^Q\left[ \rP = \vec{q}_\ell \big| \rP[j]=i \right] \V^Q\left[ b^2(\vec{q}_l) \big| \rP[j]=i \right], & \ell=\ell' \in \{m+1,\ldots, 2m \}; \cr
	\P^Q\left[ \rP = \vec{q}_\ell \big| \rP[j]=i \right] \mathrm{Cov}^Q\left[ b^2(\vec{q_\ell}), b(\vec{q_\ell}) \big| \rP[j]=i \right], & \ell\in\{1,\ldots m\}, \ell`= \ell+m; \cr
	\P^Q\left[ \rP = \vec{q}_\ell' \big| \rP[j]=i \right] \mathrm{Cov}^Q\left[ b^2(\vec{q_\ell'}), b(\vec{q_\ell'}) \big| \rP[j]=i \right], & \ell'\in\{1,\ldots m\}, \ell= \ell'+m; \cr
	0, & \mbox{otherwise,}
	\end{cases},
\end{equation}}
for $\ell, {\ell'} \in \{1,\ldots,2m\}$ and $\vec{C}' = \left( 2 \mu_{i,j} C_{i,j}(\vec{q}_{1}), \ldots,2 \mu_{i,j} C_{i,j}(\vec{q}_{m}), C_{i,j}(\vec{q}_{1}), \ldots, C_{i,j}(\vec{q}_{m}  ) \right)$.
\end{thm} 

\begin{proof}[\textbf{Proof}]
The key result to get the asymptotic distribution in this case is the multi-dimensional Anscombe's theorem,   Theorem~\ref{thm:MultiDimensionalAnscombe} in Appendix \ref{sec:appendix}. In order to apply this result we do some preparations. Using Lemma~\ref{lemma:MeasureChange}, we write
\begin{align*}
\mu_{i,j} &= \E^{\widetilde Q }[ b(\rP) | \rP[j]=i ] = \E^{ Q }[ b(\rP)  C_{i,j}(\rP)| \rP[j]=i ], 
\end{align*}
with $ C_{i,j}(\rP)= \frac{\P^{\widetilde Q}\left(\rP'=\rP \middle| \rP'[j]=i \right)}{\P^{Q}\left(\rP'=\rP \middle| \rP'[j]=i \right)}.$ Let $\vec{q}_1, \vec{q}_2, \ldots, \vec{q}_m$ be all the pairwise different paths such that {$\mbox{Supp}^Q_{i,j}.$}  Notice that in fact $b(\rP)=b(\rP, S)$ and that $\rP$ and $S$ are drawn from a product measure. Hence, it makes sense to write also $b( \vec{q})=b(\vec{q}, S)$ for a given path $\vec{q}$. For any realization $(\vec{p}_{n}, S_{n})$ of $\P^{Q}$ or $\P^{\widetilde{Q}}$  we also write $b(\vec{p}_{n})=b(\vec{p}_{n}, S_{n})$.
We can decompose $\mu_{i,j}$ using all possible paths:
\begin{equation}\label{eq:muDecomposition}
\mu_{i,j} 
= 
\sum_{\ell=1}^{m} \E^{Q} [b(\vec{q}_{\ell}) C_{i,j}(\vec{q}_{\ell}) \ind\{ \rP=\vec{q}_{\ell}\} | \rP[j]=i ].
\end{equation}

	\noindent Let us write $\tilde \mu_{\ell}=   \E^{Q} [b(\vec{q}_{\ell}) C_{i,j}(\vec{q}_{\ell}) \ind\{\rP=\vec{q}_{\ell}\} | \rP[j]=i ]$ and  $\mu_{\ell}= \E^{Q} [b(\vec{q}_{\ell}) C_{i,j}(\vec{q}_{\ell})  | \rP[j]=i ]$, and consider
	\begin{align*}
	\widehat{T}^{(n)}_{Q, \widetilde Q}(i,j) -\mu_{i,j} 
	&= \left( \frac{1}{|D^{(n)}_{(i,j)}|} \sum_{\path \in D^{(n)}_{\ij}} b(\path)\widehat{C}_{i,j}(\path, n) \right)- \mu_{i,j} \cr
	&=  \frac{1}{|D^{(n)}_{(i,j)}|}  \sum_{\ell=1}^{m} \widehat{C}_{i,j}(\vec{q}_l, n) \left( \sum_{\path \in D^{(n)}_{\ij}}  \ind\{\vec{p}=\vec{q}_{\ell}\}  b(\path) \right) -  \sum_{\ell=1}^{m} \tilde\mu_{\ell} \cr
	&=  \frac{1}{|D^{(n)}_{(i,j)}|}  \sum_{\ell=1}^{m} \widehat{C}_{i,j}(\vec{q}_\ell, n) \left( \sum_{\path \in D^{(n)}_{\ij}}  \ind\{\vec{p}=\vec{q}_\ell \}  \left(b(\vec{q}_\ell) - \E^Q[b(q_\ell)|\rP[j]=i] \right) \right) \cr
	& +  \sum_{\ell=1}^{m} \left(\widehat{C}_{i,j}(\vec{q}_\ell, n) \frac{|D^{(n)}_\ell|}{|D^{(n)}_{(i,j)}|} \E^Q[b(q_\ell)|\rP[j]=i]\right)
	-  \sum_{\ell=1}^{m} \tilde\mu_{\ell} 
	\end{align*}
	with $|D^{(n)}_{\ell}|=\sum_{k=1}^{n} \ind \{\vec{p}_k = \vec{q}_{\ell}\}$ for $\ell\in \{1,\ldots, m\}$. By independence,
	we deduce for every $\ell\in\{1,\ldots, m\}$ that
	\begin{equation}
	\begin{split}
	\tilde \mu_{\ell}
	&=   
	\E^{Q} [b(\vec{q}_{\ell}) C_{i,j}(\vec{q}_{\ell}) \ind\{\rP=\vec{q}_{\ell}\} | \rP[j]=i ] 
	\\&= 
	\E^{Q} [b(\vec{q}_{\ell})   | \rP[j]=i ] \cdot C_{i,j}(\vec{q}_{\ell})
	\cdot
	\P^Q [ \rP=\vec{q}_{\ell}  | \rP[j]=i ]
	\\&=
	\E^{Q} [b(\vec{q}_{\ell})  | \rP[j]=i ] \cdot
	\P^{\tilde{Q}} [ \rP=\vec{q}_{\ell}  | \rP[j]=i ].
	\end{split}
\end{equation}
Moreover, recalling that  $\widehat{C}_{i,j}(\vec{q}_\ell,n)
	:=
	\P^{\widetilde Q}\left(\rP=\vec{q}_\ell \middle| \rP[j]=i \right)
	\frac{ |D^{(n)}_{\ij}| }{ |D^{(n)}_\ell| }$, it follows from the previous equality that 
	\begin{equation}
		\widehat{C}_{i,j}(\vec{q}_\ell, n) \frac{|D^{(n)}_\ell|}{|D^{(n)}_{(i,j)}|} \E^Q[b(q_\ell)|\rP[j]=i]
		=
		\P^{\widetilde Q} \left[\rP=\vec{q}_\ell \middle| \rP[j]=i \right]
		\E^Q[b(q_\ell)|\rP[j]=i]
		= 
		\tilde{\mu}_\ell
	\end{equation}
for every $\ell\in\{1,\ldots, m\}$, implying that
	\begin{equation}\label{eq:TnUltimateTransformation}
		\widehat{T}^{(n)}_{Q, \widetilde Q}(i,j) -\mu_{i,j} 
		= 
		 \sum_{\ell=1}^{m} \widehat{C}_{i,j}(\vec{q}_\ell, n) \left(
		 \frac{1}{|D^{(n)}_{(i,j)}|} 
		 \sum_{\path \in D^{(n)}_{\ij}}  \ind\{\vec{p}=\vec{q}_\ell \}  \left(b(q_\ell) - \E^Q[b(q_\ell)|\rP[j]=i] \right) \right).
	\end{equation}
	Since $\ind\{\rP=\vec{q}_\ell \}$ and $ b(q_\ell)=b(q_\ell, S) $ are independent we also have that
	\begin{equation}
		\begin{split}
			\E^Q\left[\ind\{\rP=\vec{q}_\ell \}   \E^Q \left[b(q_\ell)|\rP[j]=i\right] \big| \rP[j]= i \right]
			&=
			\E^Q \left[b(q_\ell)|\rP[j]=i\right]
			\E^Q\left[\ind\{\rP=\vec{q}_\ell \}   \big| \rP[j]= i \right] 
			\\&=
			\E^Q \left[b(q_\ell)|\rP[j]=i\right]
			\E^Q\left[\ind\{\rP=\vec{q}_\ell \}   \big| \rP[j]= i \right]
			\\& = 
			\E^Q \left[\ind\{\rP=\vec{q}_\ell \} b(q_\ell)|\rP[j]=i\right]
		\end{split}
	\end{equation}
Hence,
	\begin{equation}\label{eq:zeroExpTransformOfTn}
		\E^Q\left[\ind\{\rP=\vec{q}_\ell \}   \left(b(q_\ell)  
		- \E^Q \left[b(q_\ell)|\rP[j]=i\right]  \right) \big|\rP[j]=i\right] = 0, \quad \mbox{for all } \ell= 1,\ldots, m.
	\end{equation}
We now define the random variables $Y$ in order to apply Theorem~\ref{thm:MultiDimensionalAnscombe}. In the following we assume all random variables to be distributed under $\P^{Q}$ conditioned on $\rP[j]=i$. Define for $k\in\{1,\ldots, n\}$,
	$$
	\Xi^{(k)} = \sum_{\ell=1}^{m} \ell \cdot \ind\{\rP_{k}=\vec{q}_{\ell}\}
	$$
	the variable indicating which path passing through $(i,j)$ is chosen. 
	For $\ell\in\{1,\ldots, m\}$ define 	
	\begin{equation}\label{eq:Y1tom}
		Y^{(k)}_{\ell} = \ind\{\Xi^{(k)}=\ell\} \left(b(q_\ell) - \E^Q[b(q_\ell)|\rP[j]=i]\right)
	\end{equation}
	and 
	$$
	Y^{(k)}=\left(Y^{(k)}_{1},\ldots, Y^{(k)}_{m} \right)
	$$
	which by Equation~\ref{eq:zeroExpTransformOfTn} satisfies
	$$
		\E^Q\left[ Y^{(k)} | \rP[j] = i \right] = \vec{0}.
	$$	
	Furthermore, we define
	$$
	N(n)=|D^{(n)}_{(i,j)}|.
	$$
	Then  we have that, 
	$$\frac{N(n)}n \xrightarrow[{n\rightarrow \infty}]{\P^{Q}-a.s.}\theta>0. $$ 
	We define  for $\ell\in\{1,\ldots,m\}$
	$$
	S^{(\ell)}_{N(n)} 
	= 
	\sum_{k=1}^{N(n)} Y^{(k)}_\ell.$$
Then, by Theorem~\ref{thm:MultiDimensionalAnscombe}, we get that
	\begin{equation}\label{eq:ClonclusionAnscombeOnTQ,Q'WithQUnknown}
	S_{N(n)}
	:=\frac1{\sqrt{N(n)}}
	\left(S^{(1)}_{N(n))}, \ldots, S^{(m)}_{N(n))} \right)\xrightarrow[{n\rightarrow \infty}]{\P^{Q}-distr.}
	\mathcal{N}(0, \Sigma)
\end{equation}
	where 
	\begin{equation}
	\Sigma (\ell,\ell')=
	\begin{cases}
		\P^Q\left[ \rP=\vec{q}_\ell \big| \rP[j]=i \right]
		\V^Q\left[ b\left( \vec{q}_\ell \right) \big| \rP[j] = i \right] & \ell=\ell';\\
		0 & \mbox{otherwise.}
	\end{cases}
	\end{equation} 
We also know that  $ \widehat{C}_{i,j}(\vec{q}_\ell,n)$ converges almost surely to the constant $C(\vec{q}_\ell)$ for every $\ell=1,\ldots,m$ and  therefore,
	\begin{equation}\label{eq:ChatAlmostSureConvergence}
			\left(  \widehat{C}_{i,j}(\vec{q}_1,n), \ldots,  \widehat{C}_{i,j}(\vec{q}_m,n) \right)
			\xrightarrow[n\rightarrow \infty]{\P^Q-a.s} 
			\vec{C} =
			\left(  C_{i,j}(\vec{q}_1), \ldots,  C_{i,j}(\vec{q}_m) \right).
	\end{equation}

Hence, by the multidimensional Slutsky's Theorem (Lemma 2.8 in \cite{van2000asymptotic}), the  continuous mapping theorem,  and the fact that a linear transformation of a gaussian vector is again gaussian we have by Equation~(\ref{eq:TnUltimateTransformation}) that
	\begin{equation}
	\frac{1}{\sqrt{|D^{(n)}_{\ij}|}} \left(\widehat{T}^{(n)}_{Q, \widetilde Q}(i,j) -\mu_{i,j}\right) 
	\xrightarrow[n \rightarrow \infty ]{\P-distr.}
	\mathcal{N}\left( 0, \vec{C} \cdot \Sigma \cdot \vec{C}^T \right),
	\end{equation}
which concludes the first part of the proof.
	
For the second estimator we proceed in a similar way. As in Equation~(\ref{eq:muDecomposition}), we get that
	\begin{equation}\label{eq:TildeQPartition}
		\E^{\tilde{Q}} \left[ b^2(\rP) \big| \rP[j]=i \right]
		=
		\sum_{l=1}^{m}
		\P^{\tilde{Q}}\left[ \rP = q_\ell | \rP[j]=i \right] \E^Q\left[ b^2(q_\ell) \big| \rP[j]=i \right],
	\end{equation}
	and we also see that
	\begin{equation}
		\begin{split}
			\frac{1}{|D^{(n)}_{i,j}|}
			\sum_{\path \in D^{(n)}_{i,j} }
			b^2(\path) \hat{C}(\path,n) 
			=
			&\sum_{l=1}^{m}
			\hat{C}(q_\ell,n)
			\left(\frac{1}{|D^{(n)}_{i,j}|}
			\sum_{\path \in D^{(n)}_{i,j} }
			 \ind\{\path = q_\ell \} \left( b^2(q_\ell) - \E^{Q}\left[ b^2(q_\ell) \big| \rP[j]=i \right] \right) \right)
			 \\&+
			 \sum_{l=1}^{m} \frac{|D^{(n)}_l|}{|D^{(n)}_{i,j}|}
			 \hat{C}(q_\ell,n) \E^Q\left[ b^2\left( q_\ell \right) \big| \rP[j]=i \right]
		\end{split}  
	\end{equation}
	where
	\begin{equation}\label{eq:SimplificationChat}
		\frac{|D^{(n)}_l|}{|D^{(n)}_{i,j}|}
		\hat{C}(q_\ell,n) \E^Q\left[ b^2\left( q_\ell \right) \big| \rP[j]=i \right] 
		=
		\P^{\tilde{Q}}\left[ \rP = q_\ell | \rP[j]=i \right] \E^Q\left[ b^2(q_\ell) \big| \rP[j]=i \right].
	\end{equation}
	Equations~(\ref{eq:TildeQPartition}) and (\ref{eq:SimplificationChat}) imply that
	\begin{equation}
		\begin{split}
		&\frac{1}{|D^{(n)}_{i,j}|}
		\sum_{\path \in D^{(n)}_{i,j} }
		\left(	b^2(\path) \hat{C}(\path,n)\right) 
		- 
		\E^{\tilde{Q}} \left[ b^2(\rP) \big| \rP[j]=i \right]
		=
		\\& \sum_{l=1}^{m}
		\hat{C}(q_\ell,n)
		\left(\frac{1}{|D^{(n)}_{i,j}|}
		\sum_{\path \in D^{(n)}_{i,j} }
		\ind\{\path = q_\ell \} \left( b^2(q_\ell) - \E^{Q}\left[ b^2(q_\ell,S_{\rP}) \big| \rP[j]=i \right] \right) \right)
		\end{split}
	\end{equation}
	which implies that
	\begin{equation}\label{eq:SigmahatUltimateTransformation}
		\begin{split}
			&\Sigma^{(n)}_{Q,\tilde{Q}}(i,j) - \sigma^{2}_{i,j}=
			\frac{1}{|D^{(n)}_{i,j}|}
			\sum_{\path \in D^{(n)}_{i,j} }
			\left(	b^2(\path) \hat{C}(\path,n)\right) 
			- 
			\E^{\tilde{Q}} \left[ b^2(\rP) \big| \rP[j]=i \right]
			+ 
			\left(\widehat{T}^{(n)}_{Q, \widetilde Q}(i,j)\right)^2 -\mu^2_{i,j}
			\\& =
			\sum_{l=1}^{m}
			\hat{C}(q_\ell,n)
			\left(\frac{1}{|D^{(n)}_{i,j}|}
			\sum_{\path \in D^{(n)}_{i,j} }
			\ind\{\path = q_\ell \} \left( b^2(q_\ell) - \E^{Q}\left[ b^2(q_\ell,) \big| \rP[j]=i \right] \right) \right)
			\\& + 
			\left(\sum_{\ell=1}^{m} \widehat{C}_{i,j}(\vec{q}_\ell, n) \left(
			\frac{1}{|D^{(n)}_{(i,j)}|} 
			\sum_{\path \in D^{(n)}_{\ij}}  \ind\{\vec{p}=\vec{q}_\ell \}  \left(b(q_\ell) - \E^Q[b(q_\ell)|\rP[j]=i] \right) \right)\right) \left( \widehat{T}^{(n)}_{Q, \widetilde Q}(i,j) + \mu_{i,j} \right)
		\end{split}
	\end{equation}
	where we use Equation~(\ref{eq:TnUltimateTransformation}) in the {second line of the} last equality. We are now almost ready to apply Theorem~\ref{thm:MultiDimensionalAnscombe}. In the following we assume all random variables to be distributed under $\P^{Q}$ conditioned on $\rP[j]=i$. For $\ell\in\{1,\ldots, m\}$ define 
	\begin{equation}
	W^{(k)}_{\ell} = \ind\{\Xi^{(k)}=\ell\} \left(b(q_\ell) - \E^Q[b(q_\ell)|\rP[j]=i]\right)
	\end{equation}
	where by Equation~\ref{eq:zeroExpTransformOfTn}, we have for all $\ell=1,\ldots,m$ that
	$$\E^Q\left[ W^{(k)}_\ell | \rP[j] = i \right] = 0.$$
	Furthermore, for $l=m+1,\ldots,2m$, define
	$$
	W^{(k)}_{\ell} = \ind\{\Xi^{(k)}=\ell-m\} \left( b^2(q_{\ell-m}) - \E^{Q}\left[ b^2(q_{\ell-m}) \big| \rP[j]=i \right] \right)
	$$ 
	where for all $\ell=m+1,\ldots,2m$ we get that
	$$
		\E^Q\left[ W^{(k)}_{\ell} \right]
		=
		\E^Q\left[ \ind\{\rP = q_{\ell-m}\}  b^2(q_{\ell-m}) \big| \rP[j]=i\right] - \E^{Q}\left[ b^2(q_{\ell-m}) \big| \rP[j]=i \right] \E^Q\left[\ind\{\rP = q_{\ell-m} \} \big| \rP[j]=i \right] = 0.
	$$
	We then define
	$$
		W^{(k)} = \left( W^{(k)}_{1},\ldots,W^{(k)}_{2m} \right)
	$$
	which by the previous calculations satisfies for all $k=1,\ldots,n$ that
	$$
		\E^Q\left[ W^{(k)} \big| \rP[j]=i \right] = \vec{0}.
	$$
We define  for $\ell\in\{1,\ldots,2m\}$
	$$
	S'^{(\ell)}_{N(n)} 
	= 
	\sum_{k=1}^{N(n)} W^{(k)}_\ell,$$
where  $ N(n)=|D^{(n)}_{(i,j)}|.$ 	
It follows from Theorem~\ref{thm:MultiDimensionalAnscombe} that
	\begin{equation}\label{eq:ClonclusionAnscombeOnTQ,Q'WithQUnknown}
	S'_{N(n)}
	:=\frac1{\sqrt{N(n)}}
	\left(S^{(1)'}_{N(n))}, \ldots, S^{(2m)'}_{N(n))} \right)\xrightarrow[{n\rightarrow \infty}]{\P^{Q}-distr.}
	\mathcal{N}(0, \Sigma'')
	=
	\left(S'_1,\ldots,S'_{2m}\right)
	\end{equation}
	where the covariance matrix is given by{\begin{equation}
	\Sigma'_{i,j}(\ell,\ell')=
	\begin{cases}
	\P^Q\left[ \rP = \vec{q}_\ell \big| \rP[j]=i \right] \V^Q\left[ b(\vec{q}_\ell) \big| \rP[j]=i \right], & \ell=\ell' \in \{1,\ldots,m\}; \cr
	\P^Q\left[ \rP = \vec{q}_\ell \big| \rP[j]=i \right] \V^Q\left[ b^2(\vec{q}_l) \big| \rP[j]=i \right], & \ell=\ell' \in \{m+1,\ldots, 2m \}; \cr
	\P^Q\left[ \rP = \vec{q}_\ell \big| \rP[j]=i \right] \mathrm{Cov}^Q\left[ b^2(\vec{q_\ell}), b(\vec{q_\ell}) \big| \rP[j]=i \right], & \ell\in\{1,\ldots m\}, \ell`= \ell+m; \cr
	\P^Q\left[ \rP = \vec{q}_{\ell'} \big| \rP[j]=i \right] \mathrm{Cov}^Q\left[ b^2(\vec{q}_{\ell'}), b(\vec{q}_{\ell'}) \big| \rP[j]=i \right], & \ell'\in\{1,\ldots m\}, \ell= \ell'+m; \cr
	0, & \mbox{otherwise.}
	\end{cases},
\end{equation}}
We also know that the vector 
	$$
		\vec{C}^{(n)'}=
		\left( \hat{C}_{Q,\tilde{Q}}(q_1,n), \ldots, \hat{C}_{Q,\tilde{Q}}(q_{m},n), \widehat{T}^{(n)}_{Q, \widetilde Q}(i,j) + \mu_{i,j} \right)
		\xrightarrow[n\rightarrow \infty]{\P^Q-a.s} 
		\vec{C}' :=
		\left(  C_{i,j}(\vec{q}_1), \ldots,  C_{i,j}(\vec{q}_m), 2\mu_{i,j} \right).
	$$
The remaining part of the proof is similar to the end of the first part of the proof.
\end{proof}
\section{Examples}\label{sec:examples}
The first example illustrates how we can unbias estimates if the covariates are dependent; we use simulation data to compare the estimates with the ``real'' values.\footnote{An implementation of the estimators in the language R together with the following examples can be found at \url{https://github.com/NaitsabesMue/MarkovLinearModel}} .
\begin{ex}[Unbiasing]\label{ex1}
Let us consider a DAG with $2$ lines and $2$ columns, and a random matrix $S$ with first moments
\begin{equation*}
\E[S]=\left(\begin{array}{cc}
0 & 1 \cr
 -2 & 2\end{array}  \right)
\end{equation*}
The transition $Q$ of the Markov chain are given as
\begin{equation*}
q^{(1)}(s, 1) = q^{(1)}(s, 2)=1/2,
\end{equation*}
\begin{equation*}
q^{(2)}( 1, 1)= 3/4,~q^{(2)} (1,2)=1/4,~q^{(2)}(2,1)= 1/4, \mbox{ and } q^{(2)}(2,2)= 3/4.
\end{equation*}
We calculate
\begin{equation*}
\E^{Q}[b(\path) | \path[1]=1] = \frac34 \cdot 1 + \frac14 \cdot 2 = \frac54
\end{equation*}
\begin{equation*}
 \E^{Q}[b(\path) | \path[1]=2]= \frac14 \cdot (-1) + \frac34 \cdot 0 = -\frac14 
\end{equation*}
and 
\begin{equation*}
\E^{Q}[b(\path) | \path[2]=1] = \frac34 \cdot 1 + \frac14 \cdot (-1) = \frac12 = \frac14 \cdot 2 + \frac34 \cdot 0 = \E^{Q}[b(\path) | \path[2]=2]
\end{equation*}
This implies that the biased estimators $T^{(n)}\ij$ won't be able to detect the difference in quality in the second column and would underestimate the difference in the first column.  However, Corollary \ref{cor:markovuniform} ensures that
\begin{equation*}
T^{(n)}_{Q, U}(i,j) = \frac{ 1}{2 |D^{(n)}_{(i,j)}|}  \sum_{\path \in D^{(n)}_{\ij}} \left( b(\path) \frac{ \P^{Q}\left( \rP'[j]=i \right)}{\P^{Q}\left(\rP'=\path \right)}\right)
\end{equation*}
allows to detect the differences in the second column. If we write $\path_{1}=(s, 1,1), (1,2), r),  \path_{2}=(s, (1,1), (2,2), r), \path_{3}=(s, (2,1), (1,2), r)$ and $\path_{4}=(s, (1,1), (2,2), r)$ the estimators for the second column become:



 \begin{equation*}
T^{(n)}_{Q, U}(1,2) = \frac{1}{2|D^{(n)}_{(1,2)}|} \frac12 \sum_{\path \in D^{(n)}_{(1,2)}}  \left(\frac83  b(\path) \ind\{\path=\path_{1}\}  +  8 b(\path) \ind\{\path=\path_{3}\}      \right),
\end{equation*}

\begin{equation*}
T^{(n)}_{Q, U}(2,2) = \frac{1}{2|D^{(n)}_{(2,2)}|} \frac12 \sum_{\path \in D^{(n)}_{(2,2)}}  \left(8  b(\path) \ind\{\path=\path_{2}\}  +  \frac83 b(\path) \ind\{\path=\path_{4}\}      \right).
\end{equation*}
Let us consider some data from simulations. We suppose $S$ to be a matrix with independent Gaussian entries. The means are given by $\E[S]$ above, and the variances are chosen as follows 
\begin{equation}
\V(S)=\left(\begin{array}{cc}
2 & 1 \cr
 1 & 1\end{array}  \right).
\end{equation}
We perform $n=1000$ simulation runs and compare the estimations under the assumption of uniform transition, known transitions, and estimated transitions in Tables \ref{tab:ex1E} and  \ref{tab:ex1V}. The asymptotic distributions of the estimators also allow constructing confidence intervals or tests for the difference.
\begin{table}[ht]
\begin{center}
\begin{tabular}{|c|c||c|c|c|}
\hline
 & $\E[S(1,j)] -\E[S(2,j)] $ & $T^{(n)}(1,j) - T^{(n)}(2,j)$&  $T^{(n)}_{Q, U}(1,j) - T^{(n)}_{Q, U}(2,j)$  & $\widehat{T}^{(n)}_{Q, U}(1,j) - \widehat{T}^{(n)}_{Q, U}(2,j)$ \cr
\hline 
$j=1$ & $2$ & $1.54$ &  $2.15$  & $2.01$ \cr
\hline
$j=2$ & $-1$ & $ -0.13$  & $-1.21$ & $-1.04$  \cr
\hline
\end{tabular}
\end{center}
\caption{Different estimations for the differences in mean in Example \ref{ex1}; values rounded to the third decimal, $n=1000$.}
\label{tab:ex1E}
\end{table}%

\begin{table}[ht]
\begin{center}
\begin{tabular}{|c|c||c|c|c|}
\hline
 & $\V[S(1,j)] -\V[S(2,j)] $ & $\Sigma^{(n)}(1,j) - \Sigma^{(n)}(2,j)$&  $\Sigma^{(n)}_{Q, U}(1,j) - \Sigma^{(n)}_{Q, U}(2,j)$  & $\widehat{\Sigma}^{(n)}_{Q, U}(1,j) - \widehat{\Sigma}^{(n)}_{Q, U}(2,j)$ \cr
\hline 
$j=1$ & $1$ & $0.99$ &  $1.08$  & $1.05$ \cr
\hline
$j=2$ & $0$ & $ 0.68$  & $0.13$ & $0.21$  \cr
\hline
\end{tabular}
\end{center}
\caption{Different estimations for the differences in variance in Example \ref{ex1}; values rounded to the third decimal, $n=1000$.}
\label{tab:ex1V}
\end{table}

\end{ex}

\begin{ex}[Wafer production]
Our study was motivated by a root-cause analysis in the wafer fabrication. Wafer fabrication is, in general, a procedure of many repeated sequential processes.  For instance, a simplified illustration consists of $12$ subsequent fabrication steps, see \cite{wiki:wafer}, where intermediate measurement of qualities is not feasible. Our concrete example treated up to $30$ different steps and more than $90$ machines. Unfortunately,  since our industrial partner insists on fulfilling an NDA, we cannot publish any more information about the project. Probably for the same reasons, we could not find publicly available data on other industrial projects. 
\end{ex}

In the following, we treat some textbook examples to illustrate the possible application of our method.

\begin{ex}[Tooth growth]
We consider the classical tooth growth data set, \cite{Cr:47} or \cite{R:ToothGrowth}, that studies the effect of vitamin C on tooth growth in Guinea pigs. We also refer to \cite{Gr:20} for a detailed analysis of this data set. 
  The response is the length of odontoblasts (cells responsible for tooth growth) in $60$ guinea pigs. Each animal received one of three dose levels of vitamin C ($0.5$, $1$, and $2$ mg/day) by one of two sources of vitamin C; ascorbic acid (VC) or orange juice (OJ). This experiment is balanced and corresponds to the uniform case, under the additional assumption that every path appears $10$ times.  See Figure \ref{fig:tooths1} for violin plots of the six different groups. 

\begin{figure}
\begin{center}
\includegraphics[scale=0.5]{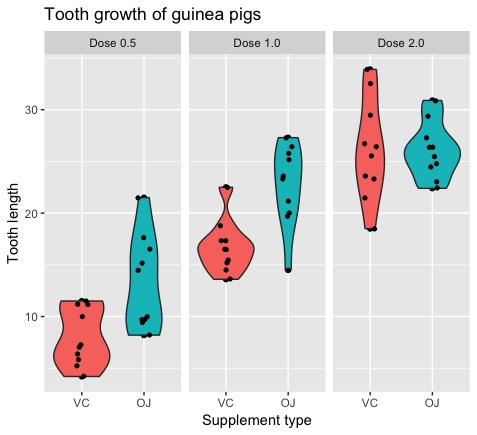}
\end{center}
\caption{Violin plot of odontoplast growth (in microns). }
\label{fig:tooths1}
\end{figure}
The figure suggests that the mean tooth growth increases with the dosage level and that OJ seems to lead to higher growth rates than VC except at a 2 mg/day dosage. The variability around the means looks to be different, however relatively small to the differences among the means. 
There might be some skew in the responses in some of the groups (e.g.,  the group OJ with dose $0.5$ has a right skew)  that may violate the standard normality assumption. 
However, a classic analysis does not reveal  statistical evidence against constance variances and  the normality assumption. The fitted coefficients in a multivariate linear model then give an increase of $3.7$ from VC to OJ and that the lengths increase with a dose of $1$mg/day (resp.\ $2$mg/day) by $9.13$ (resp. $15.49$) to the baseline of $0.5$ mg/day. These differences are reported with $p$-values smaller than $0.001$.
Without having to verify the above conditions, our method obtains the same value for the difference in mean contribution. Also, we can estimate the differences of the variances: the variance increases from VC to OJ by $23.87$. Passing from  $0.5$mg/day to $1$mg/day (resp.\ $2$mg/day) increases the variance by $0.71$ (resp.\ $5.70$).  As expected, these values coincide with the pairwise difference of the empirical variances in each group.

\end{ex}


The next examples treats a situation where variances are not constant.
\begin{ex}[Biomass response]
We consider the experiment in  \cite{Gu:13} that studied the impacts of Nitrogen (N) additions on the mass of two feather moss species (Pleurozium schreberi (PS) and Hylocomium (HS)) in an experimental forest in Sweden. More details on the classic analysis of this data set can be found in \cite{Gr:20}.  The study used a randomized block design. Here, pre-specified areas were divided into three experimental units of area  $0.1$ hectare, and one of the three treatments was randomly applied.  This procedure resulted in a balanced design with six replicates at each combination of species and treatment. 
The three treatments involved different levels of Nitrogen applied immediately after snowmelt: no additional Nitrogen (control), $12.5$ kg N per ha and year,  (N12.5), and $50$ kg N per ha and year (N50). The study's primary objective was whether the treatments would have differential impacts on the two species; while they measured additional variables, we will only observe the variables mentioned above. 
The violin plot in Figure \ref{fig:moss1} provides a first overview of the responses. 
\begin{figure}
\begin{center}
\includegraphics[scale=0.5]{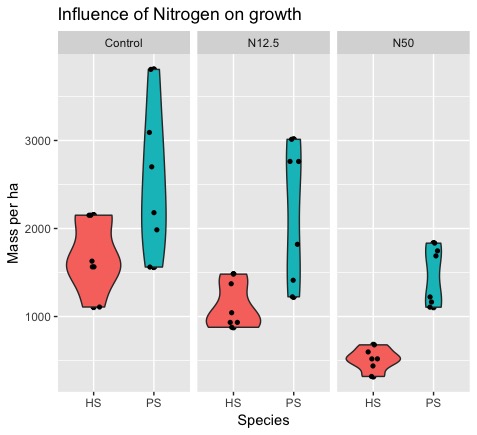}
\end{center}
\caption{Violin plot of moss growth. }
\label{fig:moss1}
\end{figure}
Figure \ref{fig:moss1} indicates some differences in variability in the different groups. Classical residual versus fitted plots reveal that there is a problem with non-constant variances. The normality assumption seems, however, to be verified. A standard remedy for non-constant variance is to transform the data to a logarithmic scale. This transformation reduces the differences in variances but induces a slight violation of the normality assumption. Our method does not demand such a transformation, and we can directly perform the pairwise comparison.  However, since the experiment is balanced, we obtain the same estimates as in the multivariate regression model: the treatment with  N12.5  (resp.\ N50) decreases the growth by    $-488.55$ (resp.\  $-1138.76$), and the species PS has in mean more $955.78$ mass grow per ha and year than the species HS.  Our method allows also quantifying the difference in variability.  The increase in variance from HS to PS is estimated with $287118.5$, and the loss of variability in N12.5 (resp.\ N50) is $29235.43$ (resp.\ $252389.76$) compared to the control group. \end{ex}

We now turn to an examples with non-constant variability and dependency between the predictors.
\begin{ex}[California test score data]
We consider the California test score data, see \cite{R:CASchools}. This data concerns  all $420$ $K-6$ and $K-8$ districts in California between $1998$ and $1999$. Test scores are on the Stanford 9 standardized test administered to 5th-grade students.  We add a new variable, called score, that is the mean of the English and math results.  This score will be our response variable.

The original data set contains many possible school characteristics or predictors for the score, as enrollment, number of teachers, number of students, number of computers per classroom, and expenditures per student.  Besides, there are demographic variables for the students that are averaged across each district. These demographic variables include the percentage of English learners, that is, students for whom English is a second language. We call this variable ``english''.
 
This dataset is an interesting textbook example since it illustrates many different aspects of multivariate regression models. In the following, we do not aim to conduct a scientific study of this data set, but only want to illustrate our method. 
A natural guess is that the student-teacher ratio may impact the pupils' test scores. For this purpose, we add a variable STR that gives the ratio between students and teachers to the dataset. The correlation between STR and scores is about $-0.22$ and can be considered statistically significant, e.g.,~the Pearson's correlation test leads to a $p$-value of less than $0.1\%$.
However, the conclusion that smaller class sizes lead to better performances might be made too fast, or their impact may be overestimated. Other variables may influence the test scores, too. The next variable that we are looking at is english; its correlation with scores is $-0.64$, which can be considered statistically significant. Moreover, STR and english correlate $0.19$. A first visual glimpse of the data can be obtained through Figure \ref{fig:CASchools1}.\begin{figure}
\begin{center}
\includegraphics[scale=0.5]{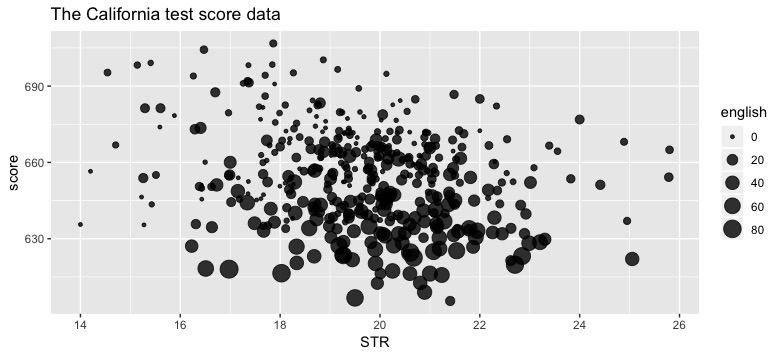}
\end{center}
\caption{Bubble plot of the score as a function of STR. The sizes of the bubbles correspond to the percentage of students for whom English is a second language.}
\label{fig:CASchools1}
\end{figure}

Our method is a priori suited to categorical predictors. We, therefore, discretize our data and create $5$ groups of equal size for STR and english. 
The groups in STR are defined by the following breaks (rounded to two decimals): $14.00$,  $18.16$, $19.27$,  $20.08$,  $21.08,$ and $25.80$. The new variable is called STRCat.
For instance, group 1 contains all observations with STR between $14.00$ and  $18.16$. The  five groups of english are defined by the  breaks:  $0.00$,  $1.16$,  $5.01$,  $13.14$, $30.72$, and $85.54$. The new variable is called englishCat.
Figure \ref{fig:CASchools2} shows violin plots of discretized data and gives a strong indication of non-constant variances.\begin{figure}
\begin{center}
\includegraphics[scale=0.5]{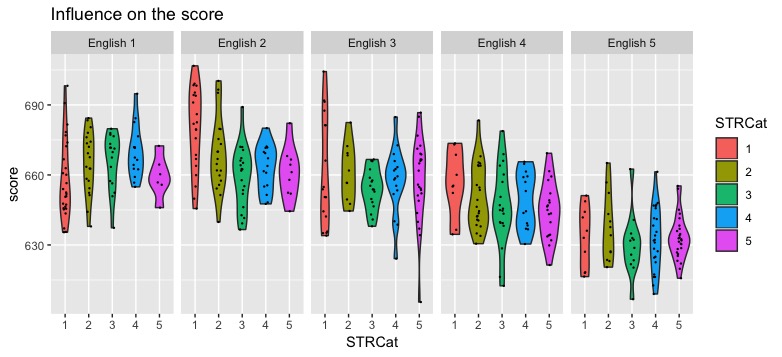}
\end{center}
\caption{Influence of student-teacher ratio and percentage of the non-native speaker on the score.}
\label{fig:CASchools2}
\end{figure}
Our Markov model can now be applied to the adapted data. Since we only have two predictors, the assumption of Markov dependencies is satisfied. We consider the variable english as the first column and the variable STRcat as the second. 

The transition of our Markov chain describes the correlation between these two categorical. Since we choose every group in englischCat of equal size $\widehat Q^{(1)}$ is the uniform distribution of $\{1,\ldots, \}$. The second Markov transition kernel is estimated by:\begin{equation}
\widehat Q^{(2)}= \left(
\begin{matrix}
  0.32 & 0.24 & 0.19 & 0.18 & 0.07 \\ 
  0.25 & 0.23 & 0.24 & 0.18 & 0.11 \\ 
  0.20 & 0.12 & 0.20 & 0.19 & 0.29 \\ 
  0.11 & 0.27 & 0.23 & 0.17 & 0.23 \\ 
  0.12 & 0.14 & 0.14 & 0.29 & 0.31 \\ 
  \end{matrix}\right).
\end{equation}
We compare the different estimators for the differences in means. Let $\Delta(i,j)=\E[S(1,j)]-\E[S(i,j)]$ for $i\in\{2, 5\}$ and $j\in\{1,\ldots\}$. We denote $\Delta^{(n)}_{U}(i,j)=T^{(n)}_{U}(1,j)-T^{(n)}_{U}(i,j)$ for the estimation under the assumption of uniform transitions and obtain
\begin{equation}
\Delta^{(n)}_{U} = \left(
\begin{matrix}
  3.91 & -3.37 \\ 
  -5.45 & -8.66 \\ 
  -14.71 & -8.81 \\ 
  -30.29 & -13.41 \\ 
  \end{matrix}\right).
\end{equation}
The unbiased estimations are given by $\widehat \Delta^{(n)}_{Q,U}(i,j)=\widehat{T}_{Q,U}^{(n)}(1,j)-\widehat{T}_{Q,U}(i,j)$:
\begin{equation}
\widehat  \Delta^{(n)}_{Q,U} = \left(
\begin{matrix}
2.68 & -1.11 \\ 
  -5.41 & -6.69 \\ 
  -14.14 & -3.34 \\ 
  -30.32 & -6.64 \\
  \end{matrix}\right).
\end{equation}
To compare the model with a multivariate regression we give the corresponding estimates  resulting from the standard linear regression:
\begin{equation}
\Delta^{(n)}_{\mathrm{regression}} = \left(
\begin{matrix}
4.39 & -0.66 \\ 
  -4.29 & -6.29 \\ 
  -13.68 & -2.88 \\ 
  -29.09 & -5.26 \\
  \end{matrix}\right).
\end{equation}
We can observe that the biased estimator $T_{U}^{(n)}$ leads to an overestimation of the effect of STR on the scores. 
Finally, let us consider the estimates for the differences in variances. For $i\in\{2, 5\}$ and $j\in\{1,\ldots\}$ let $\Gamma(i,j)=\V[S(1,j)]-\V[S(i,j)]$ and its estimators
$\Gamma^{(n)}_{U}(i,j)=\Sigma^{(n)}_{U}(1,j)-\Sigma^{(n)}_{U}(i,j)$  and $\widehat \Gamma^{(n)}_{Q,U}(i,j)=\widehat{\Sigma}_{Q,U}^{(n)}(1,j)-\widehat{\Sigma}_{Q,U}(i,j)$. We obtain
\begin{equation}
\Gamma^{(n)}_{U} = \left(
\begin{matrix}
  -64.24 & 185.54 \\ 
  -72.11 & 198.51 \\ 
  -0.56 & 149.28 \\ 
  47.12 & 198.30 \\ 
  \end{matrix}\right) \mbox{ and }
  \widehat  \Gamma^{(n)}_{Q,U} = \left(
\begin{matrix}
-68.59 & 190.41 \\ 
  -84.87 & 190.47 \\ 
  -29.32 & 196.33 \\ 
  9.40 & 244.44 \\
  \end{matrix}\right).
\end{equation} For instance, this shows that the variances in the group 1 and 5 of  english are similar and are smaller than in the other groups. The variance in group of STRcat is much higher than in the other four groups. We also see that the unbiasing leads to very different estimates on the variability.
\end{ex}

\bibliographystyle{plain}
\bibliography{bib_rca}
------------------------------------------------------------------

\appendix


\section{A multidimensional version of {Anscombe's} theorem}\label{sec:appendix}
We give a multi-dimensional version of the classical Anscombe's Theorem. The proof follows with simple modification the argument given by Renyi in his proof of Anscombe's theorem \cite{Gu:09}; it is presented here for the sake of completeness.

	\begin{thm}[Multidimensional Anscombe]\label{thm:MultiDimensionalAnscombe}
		Let $Y^{(i)}:= (Y^{(i)}_{1},Y^{(i)}_{2},\ldots,Y^{(i)}_{m})$, for $i \geq  1$, be a sequence of  i.i.d.\ real-valued random vectors with  $\E[Y^{(i)}]=0\in \R^{m}$ and covariance matrix $\Sigma$. Let $N(t)$ be a random integer-valued random variable such that $N(t)/t \xrightarrow[t\rightarrow \infty]{a.s.} \theta \in \R^+,$ then
		$$
		\frac{1}{\sqrt{N(t)}}
		\sum_{i=1}^{N(t)} Y^{(i)}
		\xrightarrow[t\rightarrow \infty]{distr.}
		\mathcal{N}(0, \Sigma ).
		$$
	\end{thm}
	\begin{proof}[\textbf{Proof}]
		Let $n (t) := \lfloor \theta t \rfloor $ and let $S_{k} := \sum_{i=1}^{k} Y^{(i)}$ and let $S^{(j)}_{k} = \sum_{i=1}^{k} Y^{(i)}_j $ then 
		\begin{equation}\label{eq:SnSplit}
		\frac{{S}_{N(t)}}{\sqrt{N(t)}}
		=
		\left(
		\left(
		\frac{S^{(1)}_{n (t)}}{ \sqrt{ n (t) }} 
		+ 
		\frac{S^{(1)}_{N (t)}- S^{(1)}_{n (t)}}{ \sqrt{ n (t) }}
		\right)
		\sqrt{\frac{n (t)}{N(t)}}
		, \ldots,
		\left(
		\frac{S^{(m)}_{n (t)}}{\sqrt{ n (t) }}
		+
		\frac{S^{(m)}_{N (t)}- S^{(m)}_{n (t)}}{ \sqrt{ n (t) }}
		\right)
		\sqrt{\frac{n (t)}{N(t)}}
		\right).
		\end{equation}
		The first observation is that, since  $n(t)$ is deterministic, due to the multi-dimensional central limit theorem we have that
		\begin{equation}\label{eq:AnscombeCLTpartNormalAscombe}
		\left(
		\frac{S^{(1)}_{ {n (t)} }}{ \sqrt{ {n (t)} }} 
		,\ldots,
		\frac{S^{(m)}_{{n (t)}}}{ \sqrt{ {n (t)} }}
		\right)
		\xrightarrow[\mbox{\tiny $t\rightarrow \infty$}]{\mbox{\tiny distr.}}
		\mathcal{N}({0}, \Sigma),
		\end{equation}
		where $\Sigma$ is the covariance matrix of the random vector $Y^{(1)}.$ 
		Next, let $\epsilon \in (0,1/3)$ be given and $n_1 (t) := \lfloor n (t) (1-\epsilon^3) \rfloor + 1$ and $n_2 (t) := \lfloor n (t) (1 + \epsilon^3) \rfloor $, then
		\begin{equation}\label{eq:UnionBoundSN-SnNormalAscombe}
		\P \left[ \bigcup_{i=1}^m \left\{ \big|S^{(i)}_{N (t)}- S^{(i)}_{n (t)} \big| > \epsilon \sqrt{n} \right\} \right]
		\leq
		\sum_{i=1}^{m}\P \left[ \big|S^{(i)}_{N (t)}- S^{(i)}_{n (t)} \big| > \epsilon \sqrt{n} \right],
		\end{equation}
		by the union bound. Let $ \sigma^2_i := \E \left[\left(Y^{(1)}_i \right)^2\right] < \infty$, then we also know that
		\begin{align*}
		\P \left[ \big|S^{(i)}_{N (t)}- S^{(i)}_{n (t)} \big| > \epsilon \sqrt{n(t)} \right]
		&=
		\P \left[ \big|S^{(i)}_{N (t)}- S^{(i)}_{n (t)} \big| > \epsilon \sqrt{n(t)} , \,  N (t) \in \left[n_1(t),n_2(t)\right] \right]
		\\&\quad +
		\P \left[ \big|S^{(i)}_{N (t)}- S^{(i)}_{n (t)} \big| > \epsilon \sqrt{n(t)} , \,  N (t) \notin \left[n_1(t),n_2(t)\right] \right]
		\\& \leq
		\P\left[ \max_{n_1(t) \leq n \leq n(t)} \big|S^{(i)}_{n}- S^{(i)}_{n (t)} \big| > \epsilon \sqrt{n(t)}  \right]
		\\&\quad +
		\P\left[ \max_{n(t) \leq n \leq n_2(t)} \big|S^{(i)}_{n}- S^{(i)}_{n (t)} \big| > \epsilon \sqrt{n(t)}  \right]
		\\&\quad +
		\P\left[ N (t) \notin \left[n_1(t),n_2(t)\right] \right]
		\\& \leq
		\frac{(n(t) - n_1(t)) \sigma^2_i }{\epsilon^2 n (t)}
		+ \frac{(n_2(t) - n(t))\sigma^2_i}{\epsilon^2 n (t)}
		\tag*{ (Kolmogorov's inequality)}
		\\&\quad+ 
		\P\left[ N (t) \notin \left[n_1(t),n_2(t)\right] \right]
		\\& \leq
		3\epsilon
		\end{align*}
		for all $i = 1,\ldots , m$ where the last inequality is valid for $t$ sufficiently large. Plugging this last estimation in  Inequality~\eqref{eq:UnionBoundSN-SnNormalAscombe} yields for $t$ sufficiently large
		\begin{align*}
		\P \left[ \bigcup_{i=i}^m \left\{ \big|S^{(i)}_{N (t)}- S^{(i)}_{n (t)} \big| > \epsilon \sqrt{n} \right\} \right]
		\leq
		3m\epsilon,
		\end{align*}
		for any $\epsilon \in (0,1/3)$. Since $\epsilon$ can be chosen arbitrarily small we deduce that
		$$
		\left(
		\frac{S^{(1)}_{N (t)}- S^{(1)}_{n (t)}}{ \sqrt{ n (t) }}
		,\ldots,
		\frac{S^{(m)}_{N (t)}- S^{(m)}_{n (t)}}{ \sqrt{ n (t) }}
		\right)
		\xrightarrow[\mbox{\tiny $t\rightarrow \infty$}]{\mbox{\tiny prob.}} (0,0,\ldots,0).
		$$
		
	By noticing  that $ \sqrt{\frac{n (t)}{N (t)}} \xrightarrow[\mbox{\tiny $t\rightarrow \infty$}]{\mbox{\tiny prob.}} 1$ and using the multidimensional version of Slutsky's theorem (Lemma 2.8 in \cite{van2000asymptotic}), we deduce that
		\begin{equation*}
		\sqrt{\frac{n (t)}{N (t)}}
		\left(
		\frac{S^{(1)}_{N (t)}- S^{(1)}_{n (t)}}{ \sqrt{ n (t) }}
		,\ldots,
		\frac{S^{(m)}_{N (t)}- S^{(m)}_{n (t)}}{\sqrt{ n (t) }}
		\right)
		\xrightarrow[\mbox{\tiny $t\rightarrow \infty$}]{\mbox{\tiny prob.}}{} 
		(0,0,\ldots,0),
		\end{equation*}
		where the last convergence is indeed in probability since it is a convergence in distribution to a constant.
		Using this last equation, Equation~\eqref{eq:SnSplit}, and the multidimensional Slutsky theorem, we conclude that
		$$
		\frac{{S}_{N(t)}}{\sqrt{N(t)}}
		\xrightarrow[\mbox{\tiny $t\rightarrow \infty$}]{\mbox{\tiny distr.}}
		\mathcal{N}({0}, \Sigma).
		$$
	\end{proof}

\section{An Anscombe version of the multivariate delta method}
We present a modification of the multivariate delta method for the case when $n$ is replaced by a random variable. The proof is a simple modification of the proof of the standard delta method. We give it for the sake of completeness.

	\begin{thm}[Anscombe's multivariate delta method]\label{thm:AscombesDeltaMethod}
		Let $\theta \in \R^k$ and $\left\{ T_n \right\}_{n \in \N}$ be a sequence of $k$ dimensional random vectors and $\left\{ \mathcal{X}_n \right\}_{n\in\N}$ be sequence of natural valued random variables such that
		\begin{equation*}
			 \begin{split}
				 &\bullet  \sqrt{\mathcal{X}_n}(T_{\mathcal{X}_n} - \theta) \xrightarrow[n\rightarrow \infty]{\mbox{\tiny distr.}} \normal_k(0,\Sigma).\\
				 &\bullet T_{\mathcal{X}_n} \xrightarrow[\tiny n\rightarrow \infty]{\mbox{\tiny prob.}} \theta
			 \end{split}
		\end{equation*}
		Furthermore, let $h:\R^k \rightarrow \R^m$ be once differentiable at $\theta$ with the gradient matrix $\nabla h(\theta)$. Then
		$$
			\sqrt{\mathcal{X}_n} \left( h(T_{\mathcal{X}_n}) - h(\theta) \right) \xrightarrow[n\rightarrow \infty]{\mbox{\tiny distr.}} \normal_k( 0,\nabla h(\theta)^T\Sigma\nabla h(\theta) ).
		$$
	\end{thm}
	
	\begin{proof}[\textbf{Proof}]
		By the definition of differentiability of a vector field, we have that 
		\begin{equation}
			h(x)
			=
			h(\theta)
			+
			(x - \theta) \cdot \nabla h (\theta)
			+\lvert x - \theta \rvert R_2 (x)			
		\end{equation}
		where $\lvert R_2(x) \rvert \xrightarrow[x \rightarrow \theta]{} 0$. In particular, we have that
		\begin{equation}\label{eq:TaylorSeries}
		\sqrt{\mathcal{X}_n}\cdot(h(T_{\mathcal{X}_n}) - h(\theta))
		=
		\sqrt{\mathcal{X}_n} \cdot (T_{\mathcal{X}_n} - \theta) \cdot \nabla h(\theta)
		+
		\left(\sqrt{\mathcal{X}_n} \cdot \lvert T_{\mathcal{X}_n} - \theta  \rvert\right) R_2(T_{\mathcal{X}_n}).
		\end{equation}
		On the other hand, it follows from the assumptions and the definition of $R_2$ that
	\begin{equation*}
		\begin{split}
			&\bullet \sqrt{\mathcal{X}_n} \cdot ( T_{\mathcal{X}_n} - \theta )
			=
			( \sqrt{\mathcal{X}_n} \left(T_{\mathcal{X}_n} - \theta \right) )  
			\xrightarrow[n\rightarrow \infty]{\footnotesize distr.} \normal_k(0,\Sigma), \\
			&\bullet R_2(T_{\mathcal{X}_n})\xrightarrow[n\rightarrow \infty]{\footnotesize prob.} 0.
		\end{split}
	\end{equation*}
Therefore, using the multidimensional Slutsky's theorem (Lemma 2.8 in \cite{van2000asymptotic}), we get that
		\begin{equation}\label{eq:VanishingResidue}
			\left(\sqrt{\mathcal{X}_n} \cdot \lvert T_{\mathcal{X}_n} - \theta  \rvert\right) R_2(T_{\mathcal{X}_n})
			\xrightarrow[n\rightarrow \infty]{\footnotesize {prob.}} 0,		
		\end{equation}
		where the last convergence is in probability because it is towards a constant. Using once more the multidimensional Slutsky's theorem together with Equations~\eqref{eq:TaylorSeries},\eqref{eq:VanishingResidue} we conclude that
		\begin{equation}
		\sqrt{\mathcal{X}_n}\cdot(h(T_{\mathcal{X}_n}) - h(\theta))
		\xrightarrow[n\rightarrow \infty]{\mbox{\tiny distr.}}
		\normal_k( 0,\nabla h(\theta)^T\Sigma\nabla h(\theta) ).
		\end{equation}
\end{proof}

\subsection*{Acknowledgment}
The authors wish to thank Alessandro  Chiancone, Herwig Friedl,  J\'er\^{o}me Depauw, and Marc Peign\'e for stimulating discussing during this project.   A.~G.~acknowledges  financial  support  from  the  Austrian  Science  Fund  project  FWF  P29355- N35.  Grateful acknowledgement is made for hospitality from TU-Graz  where the research was carried out during visits of S.~M.

\end{document}